%% file: paper.tex
\newcommand{\extendedversion}{1}

\ifnum\extendedversion=0
  \documentclass[sigconf, natbib=false]{acmart} 
\else
  \documentclass[sigconf, natbib=false, authorversion]{acmart} 
\fi

\usepackage{algorithmic}
\usepackage{graphicx}
\usepackage{textcomp}
\usepackage{xcolor}
\usepackage{listings}
\def\BibTeX{{\rm B\kern-.05em{\sc i\kern-.025em b}\kern-.08em
    T\kern-.1667em\lower.7ex\hbox{E}\kern-.125emX}}

\usepackage{xspace} 
\usepackage{suffix} 
\usepackage{comment}
\usepackage[colorinlistoftodos,prependcaption,textsize=tiny]{todonotes}
\usepackage[nolist,nohyperlinks,smaller]{acronym}
\usepackage{stfloats} 
\usepackage{lstlinebgrd}
\usepackage{cite} 
\usepackage{orcidlink}
\usepackage{tabularx} 
\usepackage{balance} 

\definecolor{lightblue}{rgb}{.651,.745,0.875}
\input{acro} 

\newboolean{our_draft}
\setboolean{our_draft}{false}

\newboolean{colored}
\setboolean{colored}{false}

\newboolean{soundness_proof}
\ifnum\extendedversion=0
  \setboolean{soundness_proof}{false}
\else
  \setboolean{soundness_proof}{true}
\fi

\input{semantics_macros}

\input{macros}

\copyrightyear{2023}
\acmYear{2023} 
\setcopyright{licensedothergov}
\acmConference[CCS '23]{Proceedings of the 2023 ACM SIGSAC Conference on Computer and Communications Security}{November 26--30, 2023}{Copenhagen, Denmark}
\acmBooktitle{Proceedings of the 2023 ACM SIGSAC Conference on Computer and Communications Security (CCS '23), November 26--30, 2023, Copenhagen, Denmark}
\acmPrice{15.00}
\acmDOI{10.1145/3576915.3623105}
\acmISBN{979-8-4007-0050-7/23/11}



\begin{document}

\ifnum\extendedversion=0
\else
  \settopmatter{printfolios=true, printacmref=false} 
\fi

\ifthenelse{\boolean{soundness_proof}}{
  \title[A Generic Methodology for the Modular Verification of Security Protocol Implementations (extended version)]{A Generic Methodology for the\\ Modular Verification of Security Protocol Implementations (extended version)}
}{
  \title[A Generic Methodology for the Modular Verification of Security Protocol Implementations]{A Generic Methodology for the\\ Modular Verification of Security Protocol Implementations}
}

\author{Linard Arquint}
\orcid{0000-0002-6230-8014}
\affiliation[obeypunctuation=true]{%
 \department{Department of Computer Science}\\
 \institution{ETH Zurich},
 \country{Switzerland}
}

\author{Malte Schwerhoff}
\orcid{0000-0003-2569-9121}
\affiliation[obeypunctuation=true]{%
 \department{Department of Computer Science}\\
 \institution{ETH Zurich},
 \country{Switzerland}
}

\author{Vaibhav Mehta}
\orcid{0000-0003-2357-3023}
\affiliation{%
 \institution{Cornell University}
 \city{Ithaca}
 \state{NY}
 \country{USA}
}
\authornote{The work was performed during a fellowship at ETH Zurich.}

\author{Peter Müller}
\orcid{0000-0001-7001-2566}
\affiliation[obeypunctuation=true]{%
 \department{Department of Computer Science}\\
 \institution{ETH Zurich},
 \country{Switzerland}
}

\input{0_abstract}

\maketitle

\input{1_introduction}
\input{2_trace_verification}
\input{3_local_state}
\input{4_security_properties}
\input{5_implementation}
\input{6_case_studies}
\input{7_trust_assumptions}
\input{8_related_work}
\input{9_conclusions}

\bibliographystyle{IEEEtran}
\balance
\bibliography{references}

\clearpage
\appendix

\ifthenelse{\boolean{soundness_proof}}{
  \input{AD_appendix}
}

\end{document}

%% file: acro.tex
\begin{acronym}
\acro{API}{application programming interface}
\acro{IV}{initialization vector}
\acro{CBC}{Ciphertext Block Chaining}
\acro{AKA}{Authentication and Key Agreement}
\acro{MITM}{Mallory-in-the-Middle}
\acro{IND-CPA}{Indistinguishability under Chosen Plaintext Attack}
\acro{LOC}{line of code}
\acro{LOS}{line of specification}
\acro{ADT}{algebraic data type}
\acro{RPC}{Remote Procedure Call}
\acro{SMT}{satisfiability modulo theories}
\acro{VST}{Verified Software Toolchain}
\acro{KDF}{key derivation function}
\acro{DH}{Diffie–Hellman}
\acro{KCI}{key compromise impersonation}
\acro{AKC}{actor key compromise}
\acro{ADT}{algebraic data type}
\acro{RVL}{reusable verification library}
\acro{CDS}{concurrent data structure}
\end{acronym}

\acrodefplural{LOC}[LOC]{lines of code}
\acrodefplural{LOC-Capital}[LOC]{Lines of code}
\acrodefplural{LOS}[LOS]{lines of specification}

%% file: macros.tex
\newcommand{\mypar}[1]{\myparbold{#1}}
\newcommand{\myparbold}[1]{\smallskip\noindent\textbf{#1.}}

\ifthenelse{\boolean{our_draft}}{
  \newcommand{\fix}[1]{\footnote{#1}}
}{
  \newcommand{\fix}[1]{}
}

\ifthenelse{\boolean{our_draft}}{
  
  \newenvironment{newtext}{\color{blue}}{}
  \newcommand{\remove}[1]{\textcolor{red}{#1}}

}{
  \ifthenelse{\boolean{colored}}{ 
    
    \newenvironment{newtext}{\color{blue}}{}
    \newcommand{\remove}[1]{}

  }{

    \newcommand{\remove}[1]{}

  }
}

\usepackage{newfloat}	
\DeclareFloatingEnvironment[
    listname={List of Protocols},
    name=Protocol,
    placement=!ht,
]{protocol}

\DeclareFloatingEnvironment[
    listname={List of Attacks},
    name=Attack,
    placement=!ht,
    within=section,
]{attack}

\usepackage{amsthm}
\newtheorem{definition}{Definition}

\newtheorem{lemma}{Lemma}
\newtheorem{theorem}{Theorem}

\renewcommand*{\proofname}{Proof sketch}

\definecolor{gobracommentcolor}{HTML}{747678}

\newcommand{\gammagray}{\textcolor{gobracommentcolor}{\gamma}}
\newcommand{\existsgray}{\textcolor{gobracommentcolor}{\exists}}

\lstdefinelanguage{gobra}{
  language=go,
  sensitive=true,
  morecomment=[l]{//},
  morecomment=[s]{/*}{*/},
  morekeywords=[1]{ 
    pred, implements, ghost, set, match, foreach, of
  },
  morekeywords=[2]{ 
    requires, ensures, invariant, req, ens, pure, unfolding, in, forall, acc, let
  },
  morekeywords=[3]{ 
    fold, unfold,
    assume, assert, inhale, exhale
  },
  basicstyle={\ttfamily\footnotesize},
  commentstyle={\color{gobracommentcolor}\textit},
  keywordstyle={[1]\color[HTML]{0005FF}},
  keywordstyle={[2]\color[HTML]{CC5500}},
  keywordstyle={[3]\color[HTML]{EC008C}},
  mathescape=true,
  moredelim=**[is][\normalfont\itshape]{'}{'}
}

\lstnewenvironment{gobra}[1][]{
  \lstset{
    language=gobra,
    floatplacement={tbp},
    captionpos=b,
    frame=lines,
    numbers=left,
    numberblanklines=true,
    xleftmargin=8pt,
    xrightmargin=8pt,
    numberstyle=\scriptsize,
    breaklines=true,
    #1
  }
}{}

\def\code{%
    \lstinline[language=gobra,basicstyle=\ttfamily]}

,keywordstyle={[2]\color{black}},keywordstyle={[3]\color{black}}]}

\lstdefinelanguage{fstar}{
  sensitive=true,
  morecomment=[l]{//},
  morecomment=[s]{/*}{*/},
  morekeywords=[1]{ 
    module, open, val, let, in, assert, fun
  },
  basicstyle={\ttfamily\footnotesize},
  commentstyle={\color[HTML]{747678}\textit},
  keywordstyle={[1]\color[HTML]{0005FF}},
  keywordstyle={[2]\color[HTML]{CC5500}},
  keywordstyle={[3]\color[HTML]{EC008C}},
  mathescape=true,
  moredelim=**[is][\normalfont\itshape]{'}{'}
}

\lstnewenvironment{fstar}[1][]{
  \lstset{
    language=fstar,
    floatplacement={tbp},
    captionpos=b,
    frame=lines,
    numbers=left,
    numberblanklines=false,
    xleftmargin=8pt,
    xrightmargin=8pt,
    numberstyle=\scriptsize,
    breaklines=true,
    #1
  }
}{}

\newcommand*{\Ie}{I.e.,\xspace}
\newcommand*{\ie}{i.e.,\xspace}
\newcommand*{\Eg}{E.g.,\xspace}
\newcommand*{\eg}{e.g.,\xspace}

\newcommand*{\cf}{cf.\xspace}

\newcommand*{\wrt}{w.r.t.\xspace}

\newcommand*{\etal}{et~al.\xspace}

\newcommand*{\term}[1]{\emph{#1}}

\newcommand*{\Figref}[1]{Fig.~\ref{fig:#1}}
\newcommand*{\figref}{\Figref}
\newcommand*{\Secref}[1]{Sec.~\ref{sec:#1}}
\newcommand*{\secref}{\Secref}
\newcommand*{\Appref}[1]{App.~\ref{sec:#1}}
\newcommand*{\appref}{\Appref}
\newcommand*{\Lineref}[1]{Line~\ref{line:#1}}
\newcommand*{\lineref}[1]{line~\ref{line:#1}}
\newcommand*{\Linerange}[2]{Lines~\ref{line:#1}--\ref{line:#2}}
\newcommand*{\linerange}[2]{lines~\ref{line:#1}--\ref{line:#2}}
\newcommand*{\Thmref}[1]{Thm.~\ref{thm:#1}}
\newcommand*{\thmref}{\Thmref}
\newcommand*{\Lemref}[1]{Lemma~\ref{lem:#1}}
\newcommand*{\lemref}{\Lemref}

\newcommand*{\Defref}[1]{Def.~\ref{def:#1}}
\newcommand*{\defref}{\Defref}
\WithSuffix{\newcommand*}\Figref*[1]{Fig.~#1}
\WithSuffix{\newcommand*}\figref*{\Figref*}
\WithSuffix{\newcommand*}\Secref*[1]{Sec.~#1}
\WithSuffix{\newcommand*}\secref*{\Secref*}
\WithSuffix{\newcommand*}\Lineref*[1]{Line~#1}
\WithSuffix{\newcommand*}\lineref*[1]{line~#1}
\WithSuffix{\newcommand*}\Lines*[1]{Lines~{#1}}
\WithSuffix{\newcommand*}\lines*[1]{lines~{#1}}
\WithSuffix{\newcommand*}\Linerange*[1]{Lines~{#1}}
\WithSuffix{\newcommand*}\linerange*[1]{lines~{#1}}

\newcommand*{\citeauthors}[2]{#1~\etal~\cite{#2}}

\makeatletter
\newcommand{\ignorecite}[1]{{\@fileswfalse\cite{#1}}}
\makeatother

\newcommand{\vb}{\vphantom{b}} 

\newcommand*{\event}[1]{\ensuremath{\mathit{#1}}}
\newcommand*{\symb}[1]{\ensuremath{\mathit{#1}}}

\newcommand{\lbbar}{\{\kern-0.5ex|} 
\newcommand{\rbbar}{|\kern-0.5ex\}} 

%% file: 0_abstract.tex
\begin{abstract}
\looseness=-1
Security protocols are essential building blocks of modern IT systems. Subtle flaws in their design or implementation may compromise the security of entire systems. It is, thus, important to prove the absence of such flaws through formal verification. Much existing work focuses on the verification of protocol \emph{models}, which is not sufficient to show that their \emph{implementations} are actually secure. Verification techniques for protocol implementations (\eg via code generation or model extraction) typically impose severe restrictions on the used programming language and code design, which may lead to sub-optimal implementations. 
In this paper, we present a methodology for the modular verification of strong security properties directly on the level of the protocol implementations. Our methodology leverages state-of-the-art verification logics and tools to support a wide range of implementations and programming languages. We demonstrate its effectiveness by verifying memory safety and security of Go implementations of the Needham-Schroeder-Lowe, Diffie-Hellman key exchange, and WireGuard protocols, including forward secrecy and injective agreement for WireGuard. We also show that our methodology is agnostic to a particular language or program verifier with a prototype implementation for C.
\end{abstract}

%% file: 1_introduction.tex
\section{Introduction}
\label{sec:introduction}

Cryptographic protocols, such as TLS, WireGuard~\cite{Donenfeld17}, and Signal~\cite{MarlinspikeP16}, are the cornerstones of today's global communication networks because they ensure crucial security properties, such as participant authentication and data privacy. With Lowe's famous attack on the Needham-Schroeder protocol~\cite{Lowe96,NeedhamS78}, it has become obvious that formal proofs are necessary for verifying that cryptographic protocols actually provide the desired properties.

The vast majority of existing work on automated protocol verification targets protocol \emph{models}, \ie abstract descriptions of the cryptographic operations and message exchanges that constitute a protocol. The verification of protocol models is useful to show the security of the protocol \emph{design}, but does not guarantee that concrete protocol \emph{implementations} are also secure. Common programming errors (\eg missing bounds checks in the Heartbleed bug~\cite{CVEHeartbleed13}) or incorrect implementations of the design (\eg accidentally omitted protocol steps in the Matrix SDK~\cite{CVEMatrix21}) may render the implementation insecure even if the protocol design is secure.

Verifying protocol implementations is substantially more complex than verifying models. Targeting realistic implementations requires reasoning, for instance, about mutable data structures, intricate control-flow (\eg dynamic dispatch), concurrency, and performance-optimized code. Moreover, implementations are significantly larger than abstract models, and change more frequently, which requires \emph{modular} verification techniques to decompose the verification task and reduce the re-verification effort when code evolves. Modular reasoning is more difficult than the non-modular analyses typically used to verify protocol models.

\looseness=-1
One approach at obtaining verified protocol implementations
is to generate an executable implementation automatically from a verified model (\eg~\cite{PozzaSD04,CadeB12,BhargavanBDHKSW21,BhargavanB0HKSW21,HoPBB22}). 
For instance, Bhargavan~\etal's DY* framework~\cite{BhargavanBDHKSW21, BhargavanB0HKSW21, HoPBB22} generates OCaml or C code from a functional implementation in F*~\cite{SwamyHKRDFBFSKZ16}. The generated code is secure by construction (provided the code generator is correct). However, changing the code manually (\eg to optimize performance) forfeits any security guarantees. To achieve modular verification, DY* relies on a specific coding discipline (at most one protocol step per F* function), which must be enforced manually, and is in general not adhered to by existing implementations. A violation of this discipline unwittingly restricts the capabilities of the attacker and, thus, may cause DY* to miss attacks.

An alternative approach is to verify security properties for a protocol model that is extracted automatically from an implementation
(\eg~\cite{BhargavanFGT08,OShea08,AizatulinGJ12,KobeissiBB17,BhargavanBK17}).
However, automatic model extraction often requires that implementations follow restrictive coding disciplines.
Similar restrictions apply to approaches based on \term{executable models} (\eg~\cite{SistoCAP18,ProtzenkoBMB19}), 
\ie models written in specific subsets of programming languages that facilitate reasoning, but typically do not provide the low-level features required for optimized implementations.

Instead of automatically generating an implementation or extracting a model, 
Arquint~\etal~\cite{ArquintWLSSWBM23} prove refinement between an \emph{existing} verified model and a corresponding \emph{existing} implementation. Their approach supports realistic implementations, but requires expertise in and relies on the soundness of two tools (a model \emph{and} a program verifier). Moreover, formal models may not always exist, or may not be in sync with an evolving implementation.

\mypar{This work}
We present a methodology for the verification of strong security properties (\eg injective agreement, forward secrecy) directly on the level of the protocol implementations. Our methodology leverages established program verification techniques that are supported by a wide range of existing automated\footnote{The proof search is automatic but relies on user-provided annotations.} tools (\eg \cite{JacobsSPVPP11,WolfACOPM21,BlomH14,SantosMNWG18,Astrauskas0PS19}), which makes it readily applicable. It is based on separation logic \cite{OHearnRY01,Reynolds02}, a program logic that supports the language features used to write efficient implementations, such as mutable heap data structures and concurrency. As a result, our methodology applies to realistic implementations written in mainstream programming languages such as C, Go, JavaScript, and Rust. Verification in our methodology is \emph{modular}, that is, one can verify each method (or protocol participant) in isolation. Modularity is crucial for scalability, to reduce the re-verification effort when the code evolves, and to provide strong guarantees for libraries.

As is common in protocol verification, we explicitly model the global trace of a protocol execution, which allows us to express security properties in ways familiar to security experts. This trace is expressed and manipulated via \emph{ghost code}~\cite{FilliatreGP14}, that is, program code that is used for verification purposes, but erased by the compiler before the program is executed. The ghost code required to manipulate the global protocol trace is encapsulated in the I/O and crypto libraries used by an implementation to ensure, \eg that each sent message is correctly reflected on the trace.

Using ghost code allows us to cleanly separate the global trace, which is necessary to prove protocol-wide properties, from the data structures maintained locally by each participant. We treat each participant instance of a protocol (including a Dolev-Yao attacker~\cite{DolevY83}) as a concurrent thread, and the global trace as shared state among these threads. This approach allows us to reason about unboundedly many participant instances and to leverage existing verification techniques and tools for shared-data concurrency.

\mypar{Contributions}
We make the following contributions:

\begin{enumerate} 
\item We present a modular verification methodology for protocol implementations, based on global traces and concurrent separation logic, that applies to a wide range of programming languages, protocol implementations, and verification tools.

\item We show how to use separation logic's linear resources to \emph{modularly} prove injective agreement, \ie the absence of replay attacks. To the best of our knowledge, we present the first invariant-based verification technique for this property.

\item We developed a reusable Go library that facilitates maintaining the global trace; protocol-independent properties are verified once and for all for this library and can, thus, be reused for different protocol implementations.

\item We demonstrate the practicality of our approach by using the Gobra verifier~\cite{WolfACOPM21} to verify memory safety and security of Go implementations of three~protocols: Needham-Schroeder-Lowe (NSL)~\cite{NeedhamS78,Lowe96}, signed Diffie-Hellman (DH)~\cite{DBLP:journals/tit/DiffieH76}, and WireGuard~\cite{Donenfeld17}. 
We show that our approach supports different programming languages and verifiers by additionally implementing a prototype of the reusable library for C and the VeriFast verifier~\cite{JacobsSPVPP11}, and using it to verify a C implementation of NSL. The implementations of our reusable verification library and the case studies are
open-source~\cite{PaperArtifact}.

\item We prove soundness of our approach, in particular, that the global trace correctly reflects all relevant protocol steps and, thus, any security property proved for the trace indeed holds for the protocol implementation.
\end{enumerate}

\noindent
We build on and substantially extend two lines of prior work:
Our use of a global trace and security labels to prove secrecy
is inspired by Bhargavan~\etal~\cite{BhargavanBDHKSW21}, but our approach achieves modularity without relying on a coding discipline (\cf earlier discussion of DY*), and thus handles existing protocol implementations soundly.
Our encoding of the global trace as a concurrent data structure is inspired by Dupressoir~\etal~\cite{DupressoirGJN11}.
Their work depends on specific features of the used programming language (\eg C's volatile fields) and verifier (VCC~\cite{CohenDHLMSST09}), while we present a separation-logic-based methodology applicable across different programming languages and verifiers, as demonstrated by our case studies.
The use of separation logic allows us to verify concurrent, heap-manipulating programs and prove security properties that so far were out of reach for invariant-based approaches.

\mypar{Outline}
\secref{trace_verification} introduces background on trace-based verification and our attacker model.
In \secref{global_trace}, we explain how we encode the global trace and how we relate it formally to the local state of each participant. 
In \secref{security_properties}, we show how to prove important security properties based on a suitable trace invariant, and how we use separation logic's linear resources to prove injective agreement.
In \secref{implementation}, we introduce our reusable verification library, which implements our methodology, and substantially reduces the verification effort per protocol. 
\secref{case-studies} describes our case studies.
We explain the trust assumptions underlying our methodology and sketch its soundness proof in \secref{trust-assumptions}, discuss related work in  \secref{related_work}, and conclude in \secref{conclusions}.

%% file: 2_trace_verification.tex
\section{Trace-based Verification}
\label{sec:trace_verification}

A protocol's security depends on the interplay of the protocol participants in the presence of an attacker. A standard technique to verify security is to record all relevant actions of the participants and the attacker on a \emph{global trace} and to formulate the intended security properties as properties of this trace. Verification then amounts to proving that all possible traces of a protocol satisfy the intended properties. In this section, we give a high-level overview of this approach; we provide the details in the later sections.

\mypar{Attacker}
We consider a Dolev-Yao attacker that has full control over the network and performs symbolic cryptographic operations.
These operations are modeled as functions over symbolic values, so-called \emph{terms}, and encode the perfect cryptography assumption, \eg that decryption succeeds if and only if it uses the correct key.

\looseness=-1
An attacker can apply these functions to all terms in its knowledge, which initially consists of all publicly-known terms, including string and integer constants.
An attacker obtains additional knowledge by reading messages on the network.
Furthermore, an attacker may corrupt participants, which adds all terms in the state of the corrupted participant to the attacker knowledge. 
We model two kinds of corruption: Corrupting a \emph{participant} leaks its long-term state, which is common to all instances of this participant, such as long-term secret keys. 
Corrupting a \emph{participant session} additionally leaks short-term state, \eg ephemeral secret keys, or exchanged nonces\footnote{Session corruption affects the entire short-term state of a participant instance, which might participate in multiple protocol sessions; a more fine-grained treatment of individual sessions is possible, but omitted for simplicity.}.

\input{3_1_sep-logic-proof-rules}

\mypar{Trace entries}
The global trace is a sequence of events. Each event corresponds to a high-level operation performed by a participant or the attacker.
It has a name and takes event-specific arguments. \Eg event \event{CreateNonce(n)} records that nonce~\symb{n} was created. This event is protocol-independent; we also support protocol-specific events to keep track of the progress within a protocol execution and to express specific security properties. \Eg a protocol-specific event may express 
which nonces or keys a participant uses to communicate with a peer
(\cf \secref{security_properties}).

We use seven protocol-independent events:
(1)~A \emph{create nonce} event records that a fresh nonce has been generated.
(2)~A \emph{send message} event records that a message has been sent on the network.
Both events may originate from a participant or the attacker.
The remaining five protocol-independent events model the capabilities of the attacker. 
(3)~The (unique) \emph{root} event is the first event on every trace and contains the initial attacker knowledge. 
(4)~An \emph{extend attacker knowledge} event models that the attacker learns additional terms, \eg by applying a cryptographic operation to a term already in the attacker knowledge.
Corruption is represented by (5)~a \emph{participant corruption} or (6)~a \emph{session corruption} event. In both cases, we use extend-events~(4) to add the newly-learned terms (from the corrupted state) to the attacker knowledge.
At any point during a protocol run, the total attacker knowledge is therefore determined by the union of the root event~(3), the send-events~(2), and the extend-events~(4).
Finally, (7)~a \emph{drop message} event records that the attacker dropped a message from the network.

\mypar{Trace invariant}
To reason modularly about the (unbounded) set of all possible traces, we introduce a \emph{trace invariant}, a property that must hold for every prefix of each trace produced by a protocol. Verification then consists of two main steps: first, proving that each action of a participant or the attacker (according to the above attacker model) maintains the trace invariant and, second, showing that the trace invariant implies the intended security properties.

An important component of a trace invariant are \emph{message invariants}, which characterize the content of a message. For instance, a message invariant might express that a message parameter is a nonce (as opposed to an arbitrary term).

%% file: 3_1_sep-logic-proof-rules.tex
\begin{figure*}
\begin{tabularx}{\textwidth}{l>{\centering\arraybackslash}Xr} 

    \Inf[\rulename{Write}]
    	{\phantom{\shoare{}{}{}}} 
    	{\shoareRes
    		{\Gamma}
    		{p \mapsto \_}
    		{\cderef{p} := e}
    		{p \mapsto e}}
    		
    &
   	
    \Inf[\rulename{Par}]
    	{\shoareRes{\Gamma}{P_1}{C_1}{Q_1}}
    	{\shoareRes{\Gamma}{P_2}{C_2}{Q_2}}
    	{\shoareRes{\Gamma}{P_1 * P_2}{C_1 \cpar C_2}{Q_1 * Q_2}}
    	
	&
	
    \Inf[\rulename{With}]
    	{\shoareRes{\Gamma}{P * I_r}{C}{Q * I_r}}
    	{\shoareRes{\Gamma,r: I_r}{P}{\cwith{r}{C}}{Q}}

\end{tabularx}
\vspace{-1.7em}
\caption{Selected separation logic proof rules: 
heap writes (\cf \secref{safety_verification}) along with parallel composition and lock-protected critical sections (\cf \secref{global-trace-encoding}).
Side-conditions are omitted for simplicity.}
\label{fig:sep-logic-proof-rules}
\vspace{-0.5em}
\end{figure*}

%% file: 3_local_state.tex
\section{Local Reasoning}
\label{sec:global_trace}

In the previous section, we have summarized how we can prove security properties based on a global trace of events. 
In this section, we show how to verify concrete protocol implementations by relating the global trace of the protocol to the local state and operations of each protocol participant. 
This verification is modular and can be automated using existing verification tools.

\subsection{Safety Verification}
\label{sec:safety_verification}

To support realistic, efficient, and existing protocol implementations, our verification technique needs to handle programming concepts such as mutable heap structures and concurrency. To this end, we employ separation logic~\cite{OHearnRY01,Reynolds02}, the de-facto standard for the modular verification of imperative code. Separation logic is supported by existing verifiers for many languages, including VeriFast~\cite{JacobsSPVPP11} for C, Prusti~\cite{Astrauskas0PS19} for Rust, and Gobra~\cite{WolfACOPM21} for Go. All of them can be used in combination with our methodology.

In separation logic, each heap location is conceptually owned by a single function execution (similar to Rust). Attempting to access a location \emph{without} owning it results in a verification failure. Ownership prevents data races in concurrent programs (since at most one function may access a location at any point in time) and facilitates reasoning about side effects (as long as one function owns a location, no other function can possibly modify it).

In specifications, the \emph{points-to assertion}~$p \mapsto e$ expresses ownership, \ie that the current function has an exclusive \emph{permission} to access location~$p$ and that $p$ has value $e$ (we write $\_$ if the value is irrelevant). For instance, the proof rule for heap updates (rule~\rulename{Write} in \figref{sep-logic-proof-rules}) enforces via its precondition that the current function execution may update $p$ only if it holds the corresponding permission.

Permissions are initially obtained when allocating a heap location, and are transferred between function executions upon call and return according to the callee function's specification. Permissions may also be transferred between threads, see \secref{global-trace-encoding}.

Verifying a protocol implementation in separation logic ensures that it is memory safe (\eg does not cause  null-pointer dereferences or buffer overflows), does not abort (\eg due to division by zero), and does not exhibit data races.
Where needed for our safety or security proof, we also verify functional correctness properties. We omit the details of safety proofs here because they are routine work in and orthogonal to the focus of this paper.

\subsection{Relating Bytes with Terms}
\label{sec:relate-bytes-with-terms}

\looseness=-1
Our global trace includes symbolic terms, such as keys, nonces, and messages. In concrete implementations, these terms are typically represented by (mutable) byte arrays. In order to relate the two, we use a \emph{concretization function} $\gamma$, which maps a term to its byte representation. We use this function in specifications; in particular, we have annotated library functions, \eg for cryptographic operations, to relate the term representations of their inputs and outputs. \Eg a hash function that maps the byte array $\mathit{xa}$ (representing, \eg a message) to the byte array $\mathit{ra}$ (representing, \eg a number) is specified by relating the corresponding terms:
$\exists x,r \ldotp \mathit{xa}=\gamma(x) \wedge \mathit{ra}=\gamma(r) \wedge r=h(x)$, where $h$ is the symbolic hash operation on terms.

Parsing a received message often requires showing that the parsed byte array $b$ corresponds to a given term $t$: $b=\gamma(t)$. Proving this property generally requires that each byte array corresponds to a
\emph{unique} term. However, this requirement is typically not satisfied in realistic implementations where, \eg a byte array of length four could store an integer or an ASCII-encoded string, which have different term representations. A possible solution~\cite{BhargavanBDHKSW21, DBLP:journals/pacmpl/0001KEW0CB20} is to enforce a unique byte-level representation for every term (for instance, by preceding it with a tag). However, this is not possible when targeting existing implementations with fixed message formats.

Therefore, we adopt a less restrictive solution here. We use the \emph{pattern requirement} from \citeauthors{Arquint}{ArquintWLSSWBM23}, which allows multiple terms to have the same byte-level representation \emph{in general}, but requires a \emph{unique} representation for the terms corresponding to protocol messages. This requirement allows a participant to uniquely determine the term for a parsed message. It ensures that the concretization function $\gamma$ is \emph{injective} on the byte arrays received as messages.
The pattern requirement is satisfied by many protocols because they include message tags to distinguish the \emph{kinds of messages}, which in turn determines the unique relationship between a byte array and a term. At the same time, it allows clashes among the representations of other terms, such as integers and strings.

\begin{figure}
\setlength{\abovedisplayskip}{0pt}
\setlength{\belowdisplayskip}{0pt}
\begin{alignat*}{3}
	& \symb{M1}.\quad & A & \rightarrow B &&:\  \{\langle 1, \symb{na}, A \rangle\}_{\symb{pk_{B}}}				\\
	& \symb{M2}.\quad & B & \rightarrow A &&:\  \{\langle 2, \symb{na}, \symb{nb}, B \rangle\}_{\symb{pk_{A}}} \\
	& \symb{M3}.\quad & A & \rightarrow B &&:\  \{\langle 3, \symb{nb} \rangle\}_{\symb{pk_{B}}}				\\
\end{alignat*}
\vspace{-2.5em} 
\caption{The NSL public key protocol, where \symb{na} and~\symb{nb} are nonces, whose generation is omitted. $\{m\}_{\symb{pk}}$ and $\langle \cdots \rangle$ denote public key encryption of plain text~$m$ under the public key~\symb{pk} and tupling, respectively. Creation and distribution of the participants' authentic keys is not part of the protocol.}
\label{fig:nsl}
\vspace{-0.5em}
\end{figure}

We illustrate the approach using the NSL public key protocol~\cite{Lowe96} in \figref{nsl}.
After receiving message \symb{M1}, Bob  parses it as an encrypted triple. The specification of the parse operation ensures $\exists \symb{na} \ldotp \gamma(m) = \gamma(\{1, \symb{na}, A\}_{\symb{pk_{B}}})$. Since $\{1, \symb{na}, A\}_{\symb{pk_{B}}}$ is a protocol message, we can apply the pattern requirement to derive the required information about $m$: $\exists \symb{na} \ldotp m = \{1, \symb{na}, A\}_{\symb{pk_{B}}}$.

\subsection{Global Trace Encoding}
\label{sec:global-trace-encoding}

As explained in \secref{trace_verification}, we use a global trace of events, verify invariants over this trace, and finally prove that the invariants imply the intended security properties. For this approach to be sound, the global trace has to include all relevant events performed by the 
protocol participants and the attacker, 
which we ensure as follows.

We model each participant instance potentially participating in a protocol session, and the attacker, as a thread in a concurrent system. Each thread maintain its own (mutable) local state, which may contain short-term, session-specific data and long-term data that is shared by all instances of the participant. Multiple instances of the same protocol role are modeled as threads that execute the same code.
Soundness of separation logic ensures that any verified property holds for all possible interleavings between the threads, that is, for all possible interactions between the participant instances and the attacker. Moreover, since separation logic is modular, it verifies the implementation of each participant in isolation, independent of the other threads potentially running in the system (assuming only that their implementations are also verified). Consequently, the verified properties hold for an arbitrary, unbounded number of participant instances.

Separation logic achieves thread-modular reasoning by ensuring that different threads operate on \emph{disjoint} memory, which prevents data races and eliminates interference between threads (see below for shared state). The proof rule for parallel composition (\rulename{Par} in \figref{sep-logic-proof-rules}) illustrates this approach. The threads $C_1$ and $C_2$ can be verified independently. They 
operate on the heap locations for which they obtain permissions via their preconditions $P_1$ and $P_2$, resp. Separation logic's \emph{separating conjunction}~$*$ in the precondition of the parallel composition expresses that the permissions in $P_1$ and $P_2$ are disjoint. Note that we show the rule for a structured parallel composition statement for simplicity; our technique also supports dynamic thread creation.

Each thread needs to manipulate its own local data structures and the global trace data structure that is shared among all threads. To support mutable shared state, we can use any of the established verification techniques for concurrency reasoning. For concreteness, we use a global \emph{lock}, which is associated with a lock invariant that needs to be established when the lock is first created. This invariant may then be assumed whenever the lock is acquired and must be proved to hold upon release.
Proof rule~\rulename{With} in \figref{sep-logic-proof-rules} illustrates this reasoning for a critical section~$C$ that is protected by the lock~$r$. $I_r$ is the invariant associated with lock~$r$, as specified by $r : I_r$ in the proof context. Conceptually, a lock owns the permissions expressed in $I_r$ and temporarily lends these permissions to a thread on entering the critical section.

Since the global trace exists only for the purpose of verification, we model it as \emph{ghost state} and all operations on it as \emph{ghost operations}; both are erased during compilation. Consequently, the lock protecting this ghost data structure can also be erased.
Reasoning about ghost locks is completely analogous to standard locks (and supported by separation logic program verifiers). However,
since a ghost lock is erased during compilation, it does not ensure mutual exclusion. Therefore, any non-ghost operation performed between an acquire and a release must be \emph{atomic} to ensure that erasing the ghost lock does not create thread interleavings that were not considered during verification.

The trace data structure provides two~operations: appending an event, and reading the current state of the global trace. \figref{overview-dynamic} illustrates how participants and the attacker interact with the global trace.
The lock invariant for the global trace is the trace invariant. By formulating this invariant in separation logic, it can express ownership of heap locations and other resources, which allows us to prove security properties that are out of reach for existing invariant-based related work, as we will see in \secref{authentication}.

\begin{figure}
\includegraphics[width=\columnwidth]{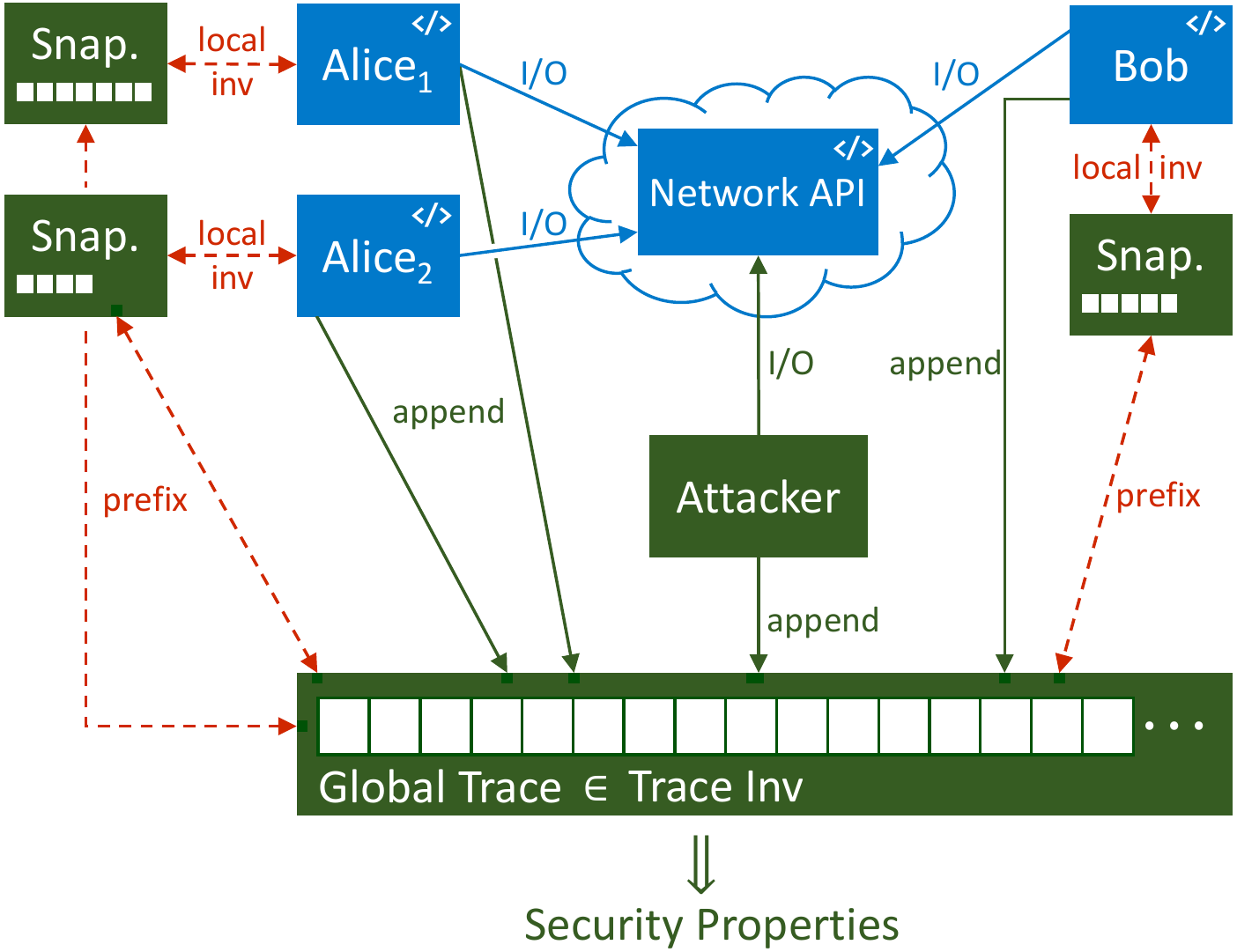}
\vspace{-1.2em}
\caption{An overview of the main components of a protocol execution in our methodology. The blue boxes are components of the protocol implementation; green boxes denote ghost structures that are used for verification. Blue and green arrows denote actual and ghost method calls, resp. The red dashed arrows denote invariants relating different data structures.
Participants and the attacker send and receive messages by interacting with the network. The attacker can perform additional I/O operations such as instructing the network to drop or modify messages. All protocol-relevant operations (including I/O operations) are recorded on a global trace. We verify (global) security properties by proving modularly that each protocol implementation (\eg two and one~implementations of Alice's and Bob's role, resp.) and the attacker maintain a trace invariant, and that the trace invariant implies the security properties.
We enable the verification of participants by relating participant-local state with the trace via local ghost state that contains a participant's local snapshot, \ie its last observed version of the trace.
}
\label{fig:overview-dynamic}
\vspace{-0.5em}
\end{figure}

Participants must record all protocol-relevant operations on the global trace. That is, to perform an operation such as sending a message or creating a nonce, they must (1)~acquire the ghost lock, (2)~perform the operation, (3)~append the corresponding event to the trace, and (4)~release the ghost lock (and at this point prove that the trace invariant is preserved). For each relevant operation, we provide a library wrapper (see \secref{implementation} for details) that performs these four steps\footnote{To avoid any runtime overhead, calls to this wrapper library could be inlined (and ghost code is erased in any case).}. Preconditions on the library functions ensure that the performed operation indeed preserves the trace invariant.
Since the trace invariant (and, hence, the preconditions) contain protocol-specific properties, our library is parametric in the invariant (\cf~\secref{implementation}). To ensure that \emph{all} relevant operations are recorded on the trace, it then suffices to perform a simple syntactic check that relevant operations are performed only via the wrapper library.

The attacker is handled similarly. We model it as code that (1)~acquires the ghost lock, (2)~determines which operations the attacker could potentially perform based on its current attacker knowledge (which is recorded on the trace), (3)~non-deterministically chooses any of these operations and appends the corresponding event to the trace, and (4)~releases the ghost lock (and at this point proves that the trace invariant is preserved). Verifying this code ensures that all possible attacker operations preserve the trace invariant.
In other words, the invariant may state only those properties that are valid under our attacker model, a property we call \emph{attacker completeness} (sometimes referred to as \emph{attacker typability}).

Participant and session corruption are two of the possible attacker operations in step~2 above. In both cases, step~3 adds all symbolic terms possibly present in the participant's (long-term or short-term) state to the attacker knowledge,  and step~4 checks that the invariant about the attacker knowledge is maintained.

\subsection{Local Snapshots}
\label{sec:local-snapshots}
\label{sec:fwdref-trace-snapshots}

To prove that a protocol-relevant operation preserves the trace invariant, we frequently need to relate the arguments of the operation to earlier events on the trace. 
For example, when sending the first message of the NSL protocol (\figref{nsl}), Alice has to show that the message invariant holds. The message invariant specifies that \symb{na} is a nonce, \ie requires a prior \event{CreateNonce(na)} event on the trace.

Discharging such proof obligations requires that participants retain information about their prior operations on the global trace. Since the global trace is a shared data structure that may grow between any two accesses, participants may soundly hold on to those facts that are \emph{stable} under extensions of the trace. For instance, if a \event{CreateNonce(na)} is present on the trace at some program point, it will also be present in all future versions of the trace.

\looseness=-1
We represent the stable information of a participant by maintaining in each participant a \emph{local snapshot} (\ie a local copy) of the global trace (see \figref{overview-dynamic}). Since the global trace may evolve by actions of other participants and the attacker, the local snapshot of a participant is generally a prefix of the global trace. Whenever a participant performs a protocol-relevant operation, we update its local snapshot to the current global trace. The trace invariant ensures that the local snapshots of all participants are prefixes of the global trace, and that these updates are the \emph{only} modifications of local snapshots.

With this design, local snapshots need to be owned by the participants (to ensure their values are retained across operations of other threads), and they must \emph{also} be owned by the ghost lock (to allow the lock invariant to relate the local snapshots to the global trace). To express this notion of shared ownership, we use fractional permissions~\cite{Boyland03}, which are supported by many separation logics. Conceptually, fractional permissions allow one to split a permission into several fractions; a non-zero fraction permits read access, whereas the full permission is required for write access. Separating conjunction \emph{adds} the fractions in both conjuncts and yields false if the sum for any location exceeds a full permission.

We split the permission to a local snapshot into two halves:
One half is part of the trace invariant and lets it express properties of the local snapshot. The other half remains with the corresponding participant and enables the participant to retain information about the global trace.
After acquiring the ghost lock, a participant temporarily obtains exclusive permission to its local snapshot by adding the half it holds with the half from the trace invariant (through the precondition 
$P * I_r$ in rule~\rulename{With} in \figref{sep-logic-proof-rules}) and can, thus, update the local snapshot.

\begin{figure}[t]
  \makeatletter
  \lst@AddToHook{OnEmptyLine}{\vspace{-0.4\baselineskip}}
  \makeatother
\begin{gobra}
na /*@, naT @*/ := CreateNonce(/*@ s @*/)$\label{line:relate-call}$
//@ assert s.NonceOccurs(naT)$\label{line:relate-local-to-global}$
\end{gobra}
\vspace{-0.7em}
\caption{Excerpt from a NSL implementation for Alice creating a nonce and demonstrating how to relate local state with the global trace. \code{//@} and \code{/*@ ... @*/} mark ghost code that is used for verification only. We omit the nonce's secrecy label (\secref{secrecy}) for simplicity.}
\label{fig:nonce-creation}
\vspace{-0.5em}
\end{figure}

By letting each participant retain a non-zero permission to its snapshot, we can rule out interference from other threads and, thus,
use standard sequential reasoning to relate the content of the local snapshot to the concrete data structures maintained by the participant (via local invariants) and to prove the presence of an event on the snapshot. The example in \figref{nonce-creation} illustrates that. \Lineref{relate-call} invokes the library function \code{CreateNonce}. Its regular result \code{na} is the generated nonce; the additional ghost result \code{naT} is the corresponding term. \code{CreateNonce} takes the caller's local snapshot \code{s} as ghost argument, which allows the function to update the snapshot and to express in its postcondition the existence of the create-nonce event on the updated local snapshot. This postcondition allows the caller to prove the assertion in \lineref{relate-local-to-global}, without having to consider any interleaving operations by other participants or the attacker.

%% file: 4_security_properties.tex
\section{Proving Security Properties}
\label{sec:security_properties}

In this section, we show how to define a trace invariant that lets us 
verify two important security properties, authentication and secrecy.
Authentication means that two protocol participants are indeed communicating with each other and (depending on the particular authentication property) agree on some common values.
Secrecy holds if confidential data remains unknown to the attacker.
While we focus here on the proof techniques for these two standard properties, our methodology is also applicable to more complex properties such as forward secrecy, as demonstrated in \secref{wireguard}.

\subsection{Authentication}
\label{sec:authentication}

To prove authentication, we use protocol-specific events to record additional information beyond the exchanged messages, so that authentication properties can be expressed in a familiar way: as correspondence between these events.  In this subsection, we show how to use trace invariants expressed in separation logic to prove two~strong and common authentication properties: non-injective and injective agreement.

We illustrate our methodology using the NSL example from  \figref{nsl}. We prove authentication using four protocol-specific events:
Before sending the first message, Alice creates event \event{Initiate(Alice, Bob, na)} to record that she wants to communicate with Bob, and use the nonce~\symb{na} in the current protocol session. After receiving the first and before sending the second message, Bob in turn creates event \event{Respond(Alice, Bob, na, nb)}, indicating the communication partners and used nonces. Finally, the events \event{FinishA} and \event{FinishB}, with the same parameters as \event{Respond}, indicate successful completion of the protocol (\ie runtime checks such as nonce comparisons succeeded) for Alice and Bob, resp. We focus on Alice's perspective in the following. We prove authentication for Bob's perspective in \secref{wireguard}, where we also discuss authentication properties for WireGuard.

\begin{figure}[t]
  \makeatletter
  \lst@AddToHook{OnEmptyLine}{\vspace{-0.4\baselineskip}}
  \makeatother
\begin{gobra}[%
	linebackgroundcolor={\ifnum\value{lstnumber}>4\ifnum\value{lstnumber}<7\color{lightblue}\fi\fi}] %
let commit = FinishA(A,B,na,nb) in
t.Occurs(commit) $\implies$ $\label{line:agreement_commit_occurs}$
let prefix, i = t.GetPrefix(commit) in
(prefix.Occurs(Respond(A,B,na,nb)) && $\label{line:agreement_running_occurs}$
  !($\exists$A$'$,B$'$,nb$'$,i$'$. i != i$'$ &&
    t.OccursAt(FinishA(A$'$,B$'$,na,nb$'$),i$'$))
) || prefix.IsCorrupted({A, B}) $\label{line:agreement_corruption}$
\end{gobra}
\vspace{-0.7em}
\caption{Non-injective (white background) and injective (all lines) agreement from Alice's perspective with Bob on the nonces~\code{na} and \code{nb}. \code{t.Occurs(e)} yields whether event~\code{e} occurs on trace~\code{t}; \code{t.GetPrefix(e)} returns \code{t}'s prefix up to and including the most recent occurrence of \code{e}, and the index of that occurrence (\ie the length of \code{prefix} minus 1). \code{t.OccursAt(e,i)} expresses that event \code{e} occurs at index~\code{i} on trace \code{t}.}
\label{fig:agreement_formula}
\vspace{-0.5em}
\end{figure}

\mypar{Non-injective Agreement}
The fact that Alice agrees with Bob on the nonces~\symb{na} and~\symb{nb}, known as \emph{non-injective agreement}~\cite{Lowe97a}, is specified in \figref{agreement_formula} (ignore the conjunct highlighted in blue for now).
This trace-based property states that if a \event{FinishA} event occurs on the trace (\lineref{agreement_commit_occurs}) then either a \event{Respond} event with matching arguments occurs earlier on the trace (\lineref{agreement_running_occurs}) or one of the participants has been corrupted before an agreement was reached (\lineref{agreement_corruption}).

To prove agreement for NSL, we include the NSL-specific property from 
\figref{nsl_event_inv} (ignore \lineref{finisha-uniqueness-witness} for now) into the trace invariant. It states that for every \event{FinishA} event, a corresponding \event{Respond} event occurred prior on the trace, or one of the participants has been corrupted. Maintaining this invariant requires us to show the occurrence of a suitable \event{Respond} event (or of corruption) when Alice creates the \event{FinishA} event.

\begin{figure}[t]
  \makeatletter
  \lst@AddToHook{OnEmptyLine}{\vspace{-0.4\baselineskip}}
  \makeatother
\begin{gobra}[%
	linebackgroundcolor={\ifnum\value{lstnumber}>2\ifnum\value{lstnumber}<4\color{lightblue}\fi\fi}] %
match ev {
  case FinishA(A, B, na, nb):
    UniWit(FinishA, na) &&$\label{line:finisha-uniqueness-witness}$
    (prefix.Occurs(Respond(A, B, na, nb)) ||
      prefix.IsCorrupted({A, B}))
  ...
}
\end{gobra}
\vspace{-0.7em}
\caption{A simplified fragment of the trace invariant for NSL-specific events. This invariant is universally quantified over the events \code{ev} occurring on the trace; \code{prefix} is the trace prefix up to event \code{ev}. The invariant expresses that whenever a \event{FinishA} event occurs on the trace, a \event{Respond} event with matching arguments must \emph{previously} occur, unless one of the participants has been corrupted.
The highlighted line includes a separation logic resource to express that the \event{FinishA} event is unique \wrt to the nonce~\symb{na}, which allows us to prove injective agreement. 
The conjunction~\code{&&} is interpreted as separation logic's separating conjunction~\code{*}.
}
\label{fig:nsl_event_inv}
\vspace{-0.5em}
\end{figure}

\looseness=-1
We discharge this proof obligation by extending the trace invariant with a message invariant for NSL's second message, which requires that the \event{Respond} event occurs on the trace or the message comes from the attacker. Hence, an implementation for Bob has to create a \event{Respond} event before sending the second message. When Alice receives the message, she may assume its message invariant (as part of the trace invariant). Since her local snapshot gets updated upon the receive-operation, the received message is recorded on the local snapshot and the message invariant becomes part of Alice's stable knowledge. So when Alice adds the \event{FinishA} event to the trace, she knows that either the \event{Respond} event occurs on the trace, or the second NSL message comes from the attacker. In the latter case, Alice can derive that corruption must have occurred because the attacker was able to construct a message containing the nonce~\symb{na}, which is accessible only to Alice and Bob (unless corrupted).

Once we established the trace invariant, it remains to show that for all traces, the invariant from \figref{nsl_event_inv} implies non-injective agreement (\figref{agreement_formula}). This proof is a standard entailment check, which is performed automatically by program verifiers.

\mypar{Injective Agreement}
The stronger property \emph{injective agreement} holds only for implementations that detect if the attacker replays messages from other protocol sessions. If successful, such a replay attack could trick participants into reusing outdated nonces (in general, key material), thereby weakening security. Proving injective agreement modularly is challenging; to the best of our knowledge, we present here the first invariant-based verification technique for injective agreement in protocol implementations (see also \secref{related_work}).

The highlighted conjunct in \figref{agreement_formula} strengthens non-injective to injective agreement by mandating that there is no second \event{FinishA} event on the trace with the same nonce \code{na}. The uniqueness of the event/nonce-pair enforces a one-to-one correspondence between \event{Respond} and \event{FinishA} events and, thus, excludes replay attacks. 

To prove injective agreement, we strengthen our trace invariant to imply this property. We could in principle include a conjunct that specifies uniqueness by quantifying over the indexes into the trace. However, such an invariant would be difficult to maintain \emph{modularly}. The necessary proof obligation for adding a \event{FinishA} event would require that no such event with the same first nonce already exists on the trace. However, each participant has only \emph{partial} information about the trace stored in its local snapshot. So even if we proved the absence of an event on the local snapshot, we could not conclude its absence on the trace, such that the proof obligation cannot be discharged.

To obtain a modular verification technique for injective agreement, we leverage separation logic's permissions to encode arbitrary linear resources (non-duplicable facts). Due to the meaning of 
separating conjunction, $p \mapsto \_  \star p \mapsto \_$ is equivalent to false (because the permissions of the two conjuncts are not disjoint). That is, the points-to assertion $p \mapsto \_$ is a \emph{non-duplicable (\ie unique) resource}. We can use this fact to model the uniqueness of an event by representing the event as a separation logic permission.
We use this mechanism as follows.

Conceptually, we tie event uniqueness to nonces because nonces are, by assumption of perfect cryptography, unique. When a protocol-specific event is declared, it can be specified as unique \wrt a specific nonce parameter. \Eg in NSL, event \event{FinishA} is unique \wrt its third parameter \symb{na}.
Subsequently, when a nonce is generated via a call to our verification library, a program annotation states for which events this nonce will be used (\eg \event{FinishA}). The library call returns not only the fresh nonce (\symb{na}), but also a linear resource for each indicated event type (technical details follow in \secref{implementation}).

This resource---called an event's \term{uniqueness witness}---then needs to be given up when the corresponding event is appended to the trace. That is, ownership of the resource is transferred from the participant to the ghost lock by conjoining the resource to the trace invariant. \Eg for NSL, Alice obtains the witness \symb{UniWit(FinishA, na)} when creating nonce \symb{na}. This witness is transferred to the trace invariant when she appends the event \event{FinishA(\_, \_, na, \_)} to the trace, as expressed by the highlighted conjunct in \figref{nsl_event_inv}. Due to the linearity of the resource, any attempt to append another \event{FinishA} event with \symb{na} would fail to verify because the required witness cannot be provided a second time, which would be necessary to preserve the trace invariant.

Consequently, the invariant from \figref{nsl_event_inv} implies that the \event{FinishA} event with \symb{na} is unique, which allows a standard separation logic verifier to prove the highlighted conjunct in the definition of injective agreement (\figref{agreement_formula}). 

Our discussion shows how the combination of a global trace and local snapshots allows us to prove authentication modularly, and how we can leverage the expressive power of separation logic to specify a trace invariant that lets us prove injective agreement.

\subsection{Secrecy}
\label{sec:secrecy}

Secrecy of a term~\symb{s}, \eg a key or a nonce, states that the attacker does not learn this term except when corrupting one of the protocol participants  that know the term. We can express secrecy as a property of the global trace because we can extract both the attacker knowledge and corruption events from the trace.

Instead of directly reasoning about the concrete attacker knowledge, we follow Bhargavan et al.~\cite{BhargavanFG10,BhargavanBDHKSW21} by over-approximating the concrete attacker knowledge to classes of terms that the attacker (possibly) knows.
This over-approximation enables modular reasoning about secrecy: we impose proof obligations that prevent secrets from being leaked to the attacker by checking for every send operation that the sent message belongs to a class already known to the attacker. For instance, if a participant tried to send a (unencrypted) secret term over the network, the send operation would be rejected by the verifier. Consequently, sending a message does not change the over-approximated attacker knowledge.
This knowledge is extended only when the attacker corrupts a participant or session.
In this case, we add the class of terms readable by the corrupted participant or session to the knowledge.

We classify terms based on their allowed recipients by assigning them a \emph{secrecy label}.
Secrecy labels range from public (\ie everyone including the attacker) over a set of participants to a set of particular protocol sessions.
The latter is useful to classify ephemeral private keys, \eg in our WireGuard case study, because only a participant running a particular protocol session is allowed to read these keys.

By proactively enforcing secrecy labels, we ensure that the (concrete) attacker knowledge may contain only public terms and terms whose secrecy label contains a participant or protocol session that is allowed to read the term and that has been corrupted in the past.
We prove this property once and for all as part of our reusable verification library (\cf \secref{implementation}).

%% file: 5_implementation.tex
 \section{Reusable Verification Library}
\label{sec:implementation}

\looseness=-1
We implement our methodology as a reusable verification library, which  significantly reduces the verification effort per protocol: the library encapsulates the global trace and provides a convenient API for common network and cryptographic operations that automates trace updates. In addition, the library provides various lemmas, such as attacker completeness (\secref{trace_verification}), which are proved once and hold for all protocols. To enable verification of a wide range of protocols, the global trace is parametric in the events it records, and the trace invariant is parametric to account for protocol-specific properties.

To demonstrate that our methodology is widely applicable, we developed a library for the Go verifier Gobra~\cite{WolfACOPM21}, and one for the C verifier VeriFast~\cite{JacobsSPVPP11}. Both library implementations are available in our open-source artifact~\cite{PaperArtifact}. In this section, we
give an overview of the library and highlight some of its technical solutions.

\begin{figure}
  \begin{center}
  \includegraphics[width=0.615\columnwidth]{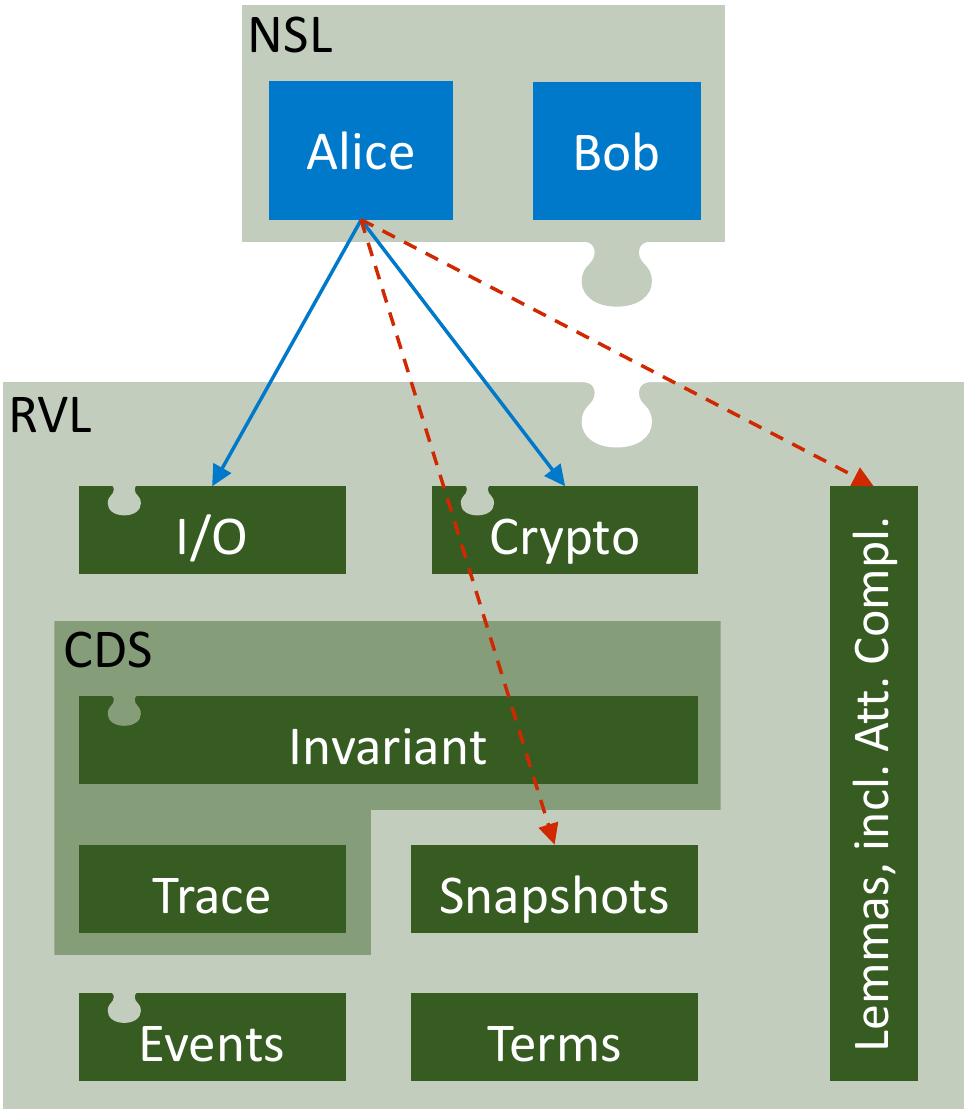}
  \end{center}
  \caption{Structure of our \acf{RVL}. The library provides implementations for the abstractions used in our methodology: terms, events, the global trace, and local snapshots. Both the trace and all local snapshots are governed by the trace invariant. The trace is encapsulated inside a \acf{CDS} that permits shared access. The APIs for I/O and cryptographic operations apply these operations and also register the corresponding events on the trace. The \ac{RVL} also provides several lemmas that have been proved for all protocols, \eg attacker completeness. Many components of the library are parametric to accommodate protocol-specific events and invariants (and the corresponding preconditions for the I/O and crypto API). We indicate parametric components using a tab symbol near the top of the box. The parameters are supplied for a concrete protocol (here, NSL), as indicated by the tab at the bottom of the box.}
  \label{fig:library}
  \vspace{-0.5em}
\end{figure}

\subsection{Overview}

In the following, we describe the library's structure and components, explain how the library can be instantiated for different protocols, and provide data on its size and verification time.

\mypar{Components}
\Figref{library} illustrates the structure of our library (lower box) and a protocol implementation that uses it (upper box). The library provides the abstractions introduced in \secref{global_trace}: terms and events abstract over concrete data structures (\eg byte arrays) and participant operations, respectively. Events are recorded on the global trace, whose content is constrained by the trace invariant. The \acf{CDS} fully encapsulates the trace, to govern shared access and maintain the invariant. Local snapshots are prefixes of the global trace, which also satisfy the trace invariant, but are owned locally by the protocol participants.

The library also provides a convenient API for common network I/O and cryptographic operations: each function performs the corresponding concrete operation (\eg sending a message or creating a nonce) and also adds the corresponding event to the trace. Suitable preconditions ensure that the operation preserves the trace invariant; they lead to proof obligations for clients using the API\@. Clients typically discharge these  with the help of stable knowledge about the trace, which is recorded in their local snapshots.

In terms of cryptographic operations, our library currently offers asymmetric encryption, authenticated encryption with associated data~(AEAD), signatures, and modular exponentiation, but can easily be extended by additional cryptographic operations.
As a reference, adding the latter two~features and proving the corresponding lemmas took about two~person days.

Note that almost the entire library consists of ghost code that is used for verification, but will be erased by the compiler. The only non-ghost operations are the calls to the underlying I/O and crypto libraries. This has two important consequences. First, these calls can be inlined in the participant implementation, such that the entire library can be removed from the executable program and does not cause any runtime overhead. Second, existing protocol implementations do not have to be modified to use the library. The library provides a convenient way to systematically annotate an implementation with ghost code and proof obligations, but other forms of annotations are also possible.

\mypar{Parametricity}
As we discussed earlier, some events and aspects of the trace invariant (and consequently the preconditions of the I/O and crypto API) are protocol-specific. To capture them, we designed our library to be parametric, such that clients using the library can instantiate it for a given protocol.  

Despite being parametric, our library nonetheless provides lemmas that are proven once and for all protocols, in particular, \emph{attacker completeness} (\secref{trace_verification}) and the \emph{secrecy lemma} (\secref{secrecy}). Attacker completeness can be proved once and for all because the library is not parametric in the kinds of term abstractions it provides. Secrecy directly follows from the protocol-independent parts of the trace invariant, which enforce for all protocols that implementations do not leak secrets to the attacker, \ie messages have to be public.
The library provides also several utility lemmas (\eg that event existence is a stable trace property) that can be used when verifying a participant implementation.

\begin{figure}[t]
  \makeatletter
  \lst@AddToHook{OnEmptyLine}{\vspace{-0.4\baselineskip}}
  \makeatother
\begin{gobra}
pred TraceInv[P](t: Trace) {
  foreach e: Entry of t:
    let pre = ... in // trace prefix up to e
    match e {
      case Send(msg):
        MsgInv[P](msg, pre)$\label{line:msgInvCall}$
      case PEvent(pe):
        P::PEventInv(pe, pre)$\label{line:eventInvCall}$
      ...
    }
}
\end{gobra}
\vspace{-0.7em}
\caption{Excerpt of the parametric trace invariant, defined via pattern matching over individual trace entries. All cases may refer to earlier events on the trace via the prefix parameter~\code{pre}.
  The case for a \event{Send} event enforces the message invariant, which is partly defined by the library, but itself parametric. 
  A \code{PEvent} represents any protocol-specific event~\code{pe}. The corresponding case of the trace invariant comes entirely from the protocol parameter \code{P}. 
}
\label{fig:trace_inv}
\vspace{-0.5em}
\end{figure}

\figref{trace_inv} shows a small excerpt of our trace invariant. The parameter \code{P} provides protocol-specific events and invariants. Besides various properties of the entire trace (not shown in the figure), the trace invariant also includes event-specific invariants. We show here the invariants for \event{Send} events and protocol-specific events. A \event{Send} event requires the message invariant, which itself can be parameterized by library clients. We prove that the generic part of the message invariant is weak enough to be preserved by the attacker; it states, in particular, that the terms occurring in the message do not leak secrets. The protocol-specific part of the message invariant may constrain only encrypted data and must allow the possibility that the encrypted data was fabricated by the attacker out of terms in the attacker knowledge. This ensures that it is maintained by all attacker actions.
For a protocol-specific event, the invariant is supplied entirely by the parameter \code{P}. In the following, we explain how this parameter is represented in our library implementations.

In the Go implementation of the library, we achieve parametricity by using Go interfaces. In particular, the \emph{generic protocol interface} declares mathematical functions (\eg \symb{isUnique} to indicate that an event is unique),
separation logic predicates (\eg protocol-specific event invariants),
and lemmas. Clients may then supply different implementations of this interface with different definitions for these functions, predicates, and lemmas. Gobra checks via suitable proof obligations that any concrete implementation satisfies key properties specified in the interface (\eg that protocol-specific invariants provide uniqueness witness resources for unique events). These properties can thus soundly be assumed while verifying the parametric library.
Analogously, parametricity \wrt events is enabled by declaring an \symb{Event} interface that protocol-specific events extend.

In VeriFast, we use its generic types (\eg for events), abstract mathematical functions (\eg \symb{isUnique}), and abstract lemmas (\eg that the event invariant is stable) to achieve parametricity and verify the library once for all protocols.
When verifying implementations of a particular protocol, these abstract functions and lemmas are concretized by providing function and lemma definitions via an automated syntactic transformation.
We prove that these definitions are not present while verifying the library, that is, we indeed verify the parametric version of the library, not a concrete instantiation.

\begin{figure}[t]
\centering
\begin{tabular}{ l c c c }
	\hline
Library & \acsp{LOC} & \acsp{LOS} & Verification time [s] \\ 
	\hline
	Go/Gobra & 83 & 6,932 & 126.1 \\
	C/VeriFast & 343 & 3,837 & 0.8 \\
	\hline
\end{tabular}
\vspace{-0.4em}
\caption{\acfp{LOC-Capital}\acused{LOC} and \acfp{LOS} (incl.\ ghost code) for the library code, together with the average verification times in Gobra and VeriFast.
}
\label{fig:library_size}
\vspace{-0.5em}
\end{figure}

\mypar{Statistics}
\figref{library_size} shows the size and verification time for the two verified implementations of our library.
As explained above, the library consists mostly of ghost code; only around 1\% is executable code.
All methods and lemmas together are verified in ca.\ 2~minutes. 
The library for VeriFast is currently less complete than the one for Gobra, and lacks several useful lemmas, which explains the smaller amount of ghost code.
It verifies in 1~second (VeriFast is usually faster than Gobra, but provides less automation).
We have measured the verification times by averaging over 30~runs on a 2020 Apple Mac mini with M1 processor and macOS Ventura 13.0.1\@. Since the library is verified once for all protocols, this effort does not have to be repeated when verifying a concrete protocol implementation.

\subsection{Technical Solutions}
\label{sec:technical-solutions}

In the following, we summarize the features of a verification technique and tool required to implement the
main abstractions (\eg terms, events, global traces) provided by our library.

\mypar{Custom Mathematical Theories}
Verification techniques frequently represent information as values of mathematical theories, such as sets, tuples, sequences, etc.\ In contrast to the corresponding data types of a programming language, these values are immutable and their operations have a direct representation in the verification logic, which simplifies reasoning. 

We use mathematical theories to represent the abstractions we use in specifications and ghost code: events, the global trace, secrecy labels, and terms with equational theories.
Conceptually, events form an \ac{ADT}, as does the global trace (a functional list). Labels and terms are also algebraic structures, but with additional properties (\eg labels have a commutative join operator).

The Gobra implementation of the library represents all these structures as uninterpreted functions with appropriate axioms (analogous to how custom theories are encoded to SMT solvers). \Eg for the ADT of events, we define axioms that ADT constructors are injective in their arguments, and that different constructors produce different events.
For terms, we define additional axioms to encode cryptographic equational theories, \eg ${g^{x\vb}}^y = {g^{y\vb}}^x$, where $g^x$ denotes Diffie-Hellman exponentiation with generator~$g$\@.
VeriFast supports ADTs natively, which we use to represent events and the global trace. For labels and terms, we again use uninterpreted functions and axioms (``auto-lemmas'') to express equational theories.

\mypar{Linear Resources}
Our novel support for proving injective agreement (\cf \secref{security_properties}) requires reasoning about the uniqueness of certain protocol-specific events. For this purpose, we introduce (ghost) memory locations and use separation logic's (exclusive) permissions to these locations as linear resources.
Separation logic predicates~\cite{parkinsonB05} allow us to construct linear resources with arbitrary parameters by mapping the parameter tuples injectively to a heap location.
We use such predicates to represent the uniqueness witnesses from \secref{security_properties}.

\mypar{Concurrency Reasoning}
As discussed in \secref{global-trace-encoding}, we model the global trace as a \emph{concurrent} data structure. Our approach is compatible with any verification technique that is able to reason about shared accesses to such a data structure and to maintain an invariant over it. Moreover, to encode local snapshots (\cf \secref{local-snapshots}), we require support for reasoning about properties that are stable under concurrency, which are offered by separation logic verifiers. 

We model the global trace as a data structure that is protected by a ghost lock. Neither Gobra nor VeriFast support ghost locks directly, but both offer standard locks. Reasoning about ghost locks and standard locks is almost identical, with one exception: Any non-ghost operations performed between acquiring and releasing a ghost lock must be atomic (because the lock will be erased by the compiler, so it does not actually provide mutual exclusion). This property is satisfied in our library.

%% file: 6_case_studies.tex
\section{Case Studies}
\label{sec:case-studies}
We applied our methodology to Go implementations of the NSL public key protocol, signed Diffie-Hellman~(DH) key exchange, and the WireGuard VPN protocol, and prove strong security properties.
We also verified a C implementation of NSL, and obtained the same security properties as for the Go implementation. Our case studies (included in our artifact~\cite{PaperArtifact})
thus demonstrate the portability of our methodology across different protocols, programming languages, and verifiers, and its scalability to realistic, interoperable  implementations.
In this section, we first summarize each of the case studies and then discuss our experiences.

\subsection{Needham-Schroeder-Lowe}
We used Gobra to verify a Go implementation of the initiator and responder roles for the NSL protocol (\cf \figref{nsl}), and likewise VeriFast for a C implementation thereof.
We implemented the core of the protocol as one method per participant; we also verified an alternative Go implementation of the initiator that contains one method per message to demonstrate that verification is not sensitive to the code structure.
Both protocol roles store their program state locally and use an invariant to relate the local state via the term abstraction to their local snapshot and, thereby, to the global trace.

\begin{figure}[t]
  \makeatletter
  \lst@AddToHook{OnEmptyLine}{\vspace{2\baselineskip}}
  \makeatother
\begin{gobra}
struct Alice  {
  SkA: byte[]
  PkB: byte[]
  Na:  byte[]
  Nb:  byte[]
  /*@ Step: uint @*/
  ...
}

/*@ pred LocalInvariant(a: Alice) {$\label{line:local-invariant-start}$
  $\existsgray$naT,nbT.
  ... && // memory omitted
  (a.Step == 2 ==>
    UniWit(FinishA, naT)) &&$\label{line:alice-uniqueness-witness}$
  (a.Step >= 2 ==>
    $\gammagray$(naT) == a.Na &&$\label{line:alice-na-gamma}$
    a.Snap().NonceOccurs(naT)) &&$\label{line:alice-na-is-nonce}$
  (a.Step >= 3 ==>
    $\gammagray$(nbT) == a.Nb &&
    a.Snap().Occurs(FinishI(A, B, naT, nbT)))$\label{line:alice-finishi-occurs}$
} @*/$\label{line:local-invariant-end}$
\end{gobra}
\vspace{-0.7em}
\caption{The struct used for Alice's local state in the Go implementation of NSL, and an excerpt from the local invariant that relates this state to Alice's local snapshot and, thereby, to the global trace. The \code{Step} field is a ghost field that is used to track Alice's progress in the protocol.}
\label{fig:local-invariant-nsl}
\vspace{-0.5em}
\end{figure}

\looseness=-1
\Figref{local-invariant-nsl} illustrates the interplay between the local state and the local snapshot for the initiator, Alice.
Alice manages her program state in a struct~\code{Alice}.
The local invariant in  \linerange{local-invariant-start}{local-invariant-end} relates Alice's local state to her local snapshot (and, thus, indirectly to the global trace).
This invariant expresses ownership of the heap locations for the struct fields, which is omitted in the figure.
More importantly, it specifies properties about the struct fields depending on Alice's progress within the protocol execution, which we keep track of via the \code{Step} field.
\Eg Alice is in Step~2 after creating the nonce~\symb{naT} and sending the first message. In this case, the invariant includes 
the uniqueness witness (\lineref{alice-uniqueness-witness}), which allows Alice to create the \event{FinishI} event in a later protocol step. The invariant relates the concrete nonce field \code{Na} to its term representation \code{naT} using the concretization function $\gamma$ (\lineref{alice-na-gamma}). This term is used in the events on the global trace. In particular, the \event{CreateNonce} event for \code{naT}  must occur on Alice's local snapshot \code{a.Snap()} (\lineref{alice-na-is-nonce}) and, thus, on the global trace.
Once Alice's protocol run has reached the final Step~3, it adds the \event{FinishI} event to the trace. The invariant reflects this  by stating that the event is on the local snapshot (\lineref{alice-finishi-occurs}).
Knowledge about \event{FinishI}'s existence on the trace entails (via the trace invariant) properties about the \event{Respond} event created by Bob (recall \figref{nsl_event_inv}). This knowledge, together with \event{FinishI}'s uniqueness witness (now stored in the trace invariant), allows us to prove injective agreement with Bob as explained in \secref{authentication}.

We prove for all participant implementations that they achieve (at the end of a protocol execution) injective agreement on, and secrecy for, both nonces~\symb{na} and \symb{nb}.
Additionally, we verify initialization code that creates an empty trace, generates public/private key pairs for the participants, and spawns two~participant instances as Go routines (similar to threads) to demonstrate that key distribution (although not part of the protocol) can be modeled using our methodology.

\subsection{Signed Diffie-Hellman}
\label{sec:dh}
In the signed DH key exchange (\cf \figref{dh}), Alice and Bob each generate a DH secret key~$x$ and $y$, respectively.
By transmitting the corresponding (signed) DH public keys~$g^x$ and $g^y$, they agree on the shared key~${g^{x\vb}}^y$ after a successful protocol run.

We prove secrecy for, and injective agreement on, the shared key.
The proof is similar to the proof for NSL, which allowed us to reuse substantial parts.
One noticable difference is that proving that both participants derive the \emph{same} shared key requires the equational theory for Diffie-Hellman exponentiation. Our reusable verification library provides such custom theories, as discussed in \secref{technical-solutions}.
Another difference is that the nonces $x$ and $y$ are not directly part of the protocol messages (in contrast to $na$ and $nb$ in NSL), but are existentially quantified in the message invariants.
A participant instance can determine the values of these existentially-quantified variables after receiving a protocol message, by connecting the message invariant to its own Diffie-Hellman secret key.

\begin{figure}
\setlength{\abovedisplayskip}{0pt}
\setlength{\belowdisplayskip}{0pt}
\begin{alignat*}{3}
	& \symb{M1}.\quad & A & \rightarrow B &&:\  g^x				\\
	& \symb{M2}.\quad & B & \rightarrow A &&:\  \lbbar\langle 0, B, A, g^x, g^y \rangle\rbbar_{\symb{sk_{B}}} \\
	& \symb{M3}.\quad & A & \rightarrow B &&:\  \lbbar\langle 1, A, B, g^y, g^x \rangle\rbbar_{\symb{sk_{A}}}				\\
\end{alignat*}
\vspace{-2.5em} 
\caption{The signed DH key exchange protocol, where $g^x$ and $g^y$ are DH public keys and $\lbbar m \rbbar_{\symb{sk}}$ denotes cryptographically signing a payload~$m$ with a secret key~$sk$.}
\label{fig:dh}
\vspace{-0.5em}
\end{figure}

\subsection{WireGuard}
\label{sec:wireguard}

As our main case study, we have picked the WireGuard VPN protocol as a real-world protocol achieving even stronger security properties than NSL\@.
WireGuard is a modern, open-source, and cross-platform VPN that uses state-of-the-art cryptography and is part of the Linux kernel. The WireGuard protocol, which performs an authenticated key exchange, has been analyzed rigorously~\cite{DowlingP18,LippBB19}.
It consists of a handshake and transport phase.
During the handshake phase, the protocol participants agree on two session keys~\symb{k_{IR}} and~\symb{k_{RI}}, one per direction, that are used to symmetrically encrypt VPN packets in the transport phase.

\mypar{Implementation}
\looseness=-1
We used the existing Go implementation from Arquint et~al.~\cite{SoundVerificationWireGuardSources}, whose memory safety proof we reused.
Thanks to our reusable verification library's parametric design, instantiating our library with the concrete networking library used by the WireGuard implementation was straightforward and only required annotating cryptographic functions with suitable postconditions.

Arquint et~al.'s implementation is a subset of WireGuard's official Go implementation. It omits advanced VPN features such as DDoS protection, session key renewal, and support for multiple concurrent VPN connections. Moreover, their implementation reduces concurrency (which we partly re-introduced, as we discuss below), and replaces a message buffer pool by single-use buffers.
Our technique could handle the removed features with additional effort that is mostly orthogonal to our methodology. For instance, the implementation of DDoS protection collects metrics (which does not pose a challenge for program verification) and uses a slightly different handshake (whose verification is analogous to the standard handshake; the differences are not relevant for authenticity and secrecy). Supporting multiple VPN connections requires slightly more complex data structures, as do buffer pools, which can easily be handled in separation logic.
Despite these simplifications, the implementation is interoperable with other WireGuard implementations and supports tunneling IP packets via the established VPN connection to and from the operating system.
Since each IP packet is encrypted using a distinct counter value, a new handshake must be performed before the counter reaches its upper limit, which is not yet implemented.
Instead, the implementation stops forwarding IP packets at that point.

Our case study goes substantially beyond of Arquint et~al.'s, which focuses on connecting Tamarin to code-level verification, and proves weak forward secrecy and non-injective agreement in the presence of long-term key corruption.
We additionally consider session corruption, \ie the possibility for an attacker to obtain ephemeral key material, and prove \emph{strong} forward secrecy and \emph{injective agreement with actor key compromise (AKC) security}.
Furthermore, we have re-introduced (from the official WireGuard implementation) and verified the ability to send and receive transport messages in the initiator \emph{concurrently}.
This change increases TCP~throughput compared to Arquint et~al.'s implementation by a factor of 180, which illustrates how important such code optimizations are for real-world protocol implementations. The initiator verified in our work reaches 72\% of the official implementation's throughput; the additional concurrency needed to close the remaining performance gap,
requires standard concurrency reasoning in separation logic, which is supported by our methodology.

\mypar{Security Properties}
Since the session keys are based on ephemeral as well as long-term key material that is contributed by both protocol participants, WireGuard achieves strong security properties.
In particular, we prove forward secrecy and injective agreement, both with \ac{AKC} security.
While WireGuard optionally incorporates a pre-shared symmetric key into the handshake to increase security, we prove all security properties in this section without considering this pre-shared key, \ie we treat the pre-shared key as a term known to the attacker.
In the following, we call the initiator \emph{actor} and the responder \emph{peer} when proving a property from the initiator's perspective, and vice versa for the responder's perspective.

Forward secrecy protects sessions against future corruption of the long-term secret keys.
\Ie an attacker cannot compute the session keys of an already established session after learning the long-term secret keys.
However, sessions that get established after corrupting the long-term secret keys are not protected because the attacker can impersonate participants by knowing their secret keys.
The literature distinguishes between weak and strong forward secrecy. We were able to reuse formalizations from existing work~\cite{HoPBB22,GirolHSJCB20,DBLP:conf/ccs/CremersHHSM17}, which are phrased as trace-based security properties and, thus, directly supported by our methodology.

\begin{figure}[t]
  \makeatletter
  \lst@AddToHook{OnEmptyLine}{\vspace{-0.4\baselineskip}}
  \makeatother
\begin{gobra}
!t.AttackerKnows(s) ||$\label{line:pfs-unknown}$
  t.GetHs(ASess, PSess).IsCorrupted({$\colorbox{lightblue}{A,}$ P}) ||$\label{line:pfs-participant-corruption}$
  t.IsSessionCorrupted({ASess, PSess})$\label{line:pfs-session-corruption}$
\end{gobra}
\vspace{-0.7em}
\caption{Strong (without highlighted part) and weak forward secrecy (entire property) for a session key~\code{s} on trace~\code{t}.
	\code{A} and \code{P} identify the actor and peer that derive the key in their protocol sessions \code{ASess} and \code{PSess}, respectively.
	\code{t.GetHs(ASess, PSess)} returns a prefix of \code{t} up to and including the corresponding handshake's completion from the actor's perspective.
	The key is protected against (future) participant corruption after the handshake's completion.}
\label{fig:pfs}
\vspace{-0.5em}
\end{figure}

\emph{Weak forward secrecy} for a session key~\symb{s} (\cf entire \figref{pfs}) holds if at any point in time, one of the following three properties hold: (1)~The attacker does not know \symb{s} (\lineref{pfs-unknown}), (2)~the actor or its peer has been corrupted before completing the handshake (\lineref{pfs-participant-corruption}), or (3)~the actor's or peer's session has been corrupted (\lineref{pfs-session-corruption}).
In the last case, the attacker gets to read the long-term and short-term state of the corrupted participant, that is, the long-term secret key and also the session keys if the session is established.
Hence, the attacker either directly obtains the session keys if the session is already established or otherwise uses the long-term secret key to impersonate the actor or its peer while establishing a session in the future. The session keys of all other sessions remain secret.

Compared to weak forward secrecy, session keys satisfying \emph{strong forward secrecy} are additionally protected against corrupting the actor, \ie the highlighted actor is removed from \lineref{pfs-participant-corruption} in \figref{pfs}.
In particular, having access to the actor's long-term secret key does not allow the attacker to obtain the established session keys.
This resilience has been formalized as \acf{AKC} by Basin et al.~\cite{DBLP:conf/csfw/BasinCH14}, generalizing the more widely known notion of \ac{KCI}.

From the initiator's perspective, WireGuard guarantees strong forward secrecy for the two~session keys once the handshake has been completed. In contrast, the responder guarantees only weak forward secrecy by the end of the handshake, but achieves strong forward secrecy after receiving the first transport message. We verified strong forward secrecy at the appropriate points in the protocol for both roles.

The responder's forward secrecy guarantee is strengthened by receiving and successfully processing the first transport message because this message acts as a key confirmation.
\Ie the responder checks that it derived the same session key~\symb{k_{IR}} as the initiator, which allows the responder to detect \ac{AKC} attacks.
Based on strong forward secrecy for the session keys, we further prove that the VPN payloads are treated with the same level of secrecy. This induces proof obligations that a participant sends VPN payloads to the network in a way that they can be read only by participants allowed to read the session keys (\eg by encrypting the VPN payloads with one of the session keys).

\begin{figure}[t]
  \makeatletter
  \lst@AddToHook{OnEmptyLine}{\vspace{-0.4\baselineskip}}
  \makeatother
\begin{gobra}[%
	linebackgroundcolor={\ifnum\value{lstnumber}>5\ifnum\value{lstnumber}<8\color{lightblue}\fi\fi}] %
let commit = Commit(A,P,ASess,PSess,m) in
let running = Running(A,P,ASess,PSess,m) in
t.Occurs(commit) $\implies$
let prefix, i = t.GetPrefix(commit) in
(prefix.Occurs(running) &&
  !($\exists$A$'$,P$'$,ASess$'$,PSess$'$,i$'$. i != i$'$ &&
    t.OccursAt(Commit(A$'$,P$'$,ASess$'$,PSess$'$,m),i$'$))
) || prefix.IsCorrupted({P})
  || prefix.IsSessionCorrupted({ASess})
\end{gobra}
\vspace{-0.7em}
\caption{Injective agreement with \acs{AKC} security on a term~\code{m} from the actor~\code{A}'s perspective with a peer~\code{P}. The highlighted conjunct indicates the \code{Commit}~event's uniqueness requirement for the given \code{m}.}
\label{fig:inj-agreement-with-akc}
\vspace{-0.5em}
\end{figure}

Confirming the session keys not only enables strong forward secrecy for the session keys but also provides additional authentication guarantees:
\emph{Injective agreement with \ac{AKC} security} (\cf \figref{inj-agreement-with-akc}) states that (1) an actor~\symb{A} agrees with a peer~\symb{P} on a term~\symb{m} with a one-to-one correspondence between the \event{Commit} and \event{Running} events unless (2) the actor's session or (3) the peer's (short-term or long-term) state has been corrupted.
In particular, corrupting the actor is not sufficient to satisfy this property.
In contrast, the NSL protocol only satisfies injective agreement \emph{without} \ac{AKC} security (as presented in \secref{authentication}) from the initiator's perspective because having access to the initiator's secret key enables the attacker to decrypt the second message, obtain the nonces~\symb{na} and \symb{nb}, and construct a modified second message containing \symb{na} and \symb{nb'} with $\symb{nb} \neq \symb{nb'}$.
Thus, there is no correspondence between \event{Commit} and \event{Running} events in the case of actor key compromise because the initiator and responder do not agree on the nonces.

\subsection{Discussion}
\label{sec:case-studies-discussion}

\begin{figure}[t]
\centering
\begin{tabular}{@{\hspace*{1.8em}} l c c c }
	\hline
	\multicolumn{1}{l}{Case Studies} & \acsp{LOC} & \acsp{LOS} & Verification time [s] \\
	\hline
	\multicolumn{1}{l}{Go/Gobra} \\
	NSL & 197 & 924 & 97.6 \\ 
	Signed DH & 225 & 888 & 119.0 \\ 
	WireGuard & 557 & 5,815 & 268.1 \\ 
	\hline
	\multicolumn{1}{l}{C/VeriFast} \\
	NSL & 300 & 1,014 & 5.0 \\ 
	\hline
\end{tabular}
\vspace{-0.4em}
\caption{\acfp{LOC-Capital}\acused{LOC} and \acfp{LOS} (incl.\ ghost code) for our case studies, together with the average verification times in Gobra and VeriFast. We performed the measurements in the same way as in \figref{library_size}.
}
\label{fig:case_studies_size}
\vspace{-0.5em}
\end{figure}

For each case study, \figref{case_studies_size} reports the size of the implementation and its specification, along with the verification time.
We exclude the alternative NSL initiator implementation in Go, and the reusable verification library (recall \figref{library_size}).
However, the specifications do include the (ghost) code instantiating our reusable verification library:
it amounts to 374, 370, and 1,077~\acp{LOS} in Gobra for NSL, DH, and WireGuard, respectively, and 391~\acp{LOS} in VeriFast for NSL\@.

Overall, the annotation overhead for Gobra ranges between 3.9 and 10.4~\acp{LOS} per line of code, and is in the typical range for modular program verification. For example, Wolf~\etal~\cite{WolfACOPM21} report a ratio of~2.7 for a small example using concurrency in Gobra.
VST-Floyd~\cite{CaoBGDA18}, a separation logic-based verifier based on Coq, reports an average ratio of~13.9 for small C programs.
Both works verify only memory safety and functional properties, but do not include any (arguably much more complex) security properties, whereas our numbers include safety and security.
As another data point, Arquint~\etal~\cite{ArquintWLSSWBM23} prove security properties for WireGuard in Gobra with a ratio of~6.5, in addition to a Tamarin model of 350~lines and a Tamarin oracle implemented in Python, which ensures that Tamarin’s proof search terminates.
Counting the Tamarin model and oracle as specification, the overall ratio is~7.1 (and requires the use of three different languages).

The main challenge in our case studies was finding
a sufficiently strong trace invariant to prove the presented security properties.
For WireGuard, we had to find suitable message invariants such that the secrecy labels for the derived session keys~\symb{k_{IR}} and~\symb{k_{RI}} are sufficiently strong to prove weak and strong forward secrecy.
These secrecy labels are related to the message invariants because the session keys are derived by an eightfold~application of \acp{KDF} that factors in long-term and ephemeral, \ie session-specific, Diffie-Hellman key material that is either locally generated or received from the peer.
Thus, each \ac{KDF} application results in a new key with a secrecy label that depends on the secrecy labels of the input key material.
To keep the annotations related to the secrecy labels in the implementation to a minimum, we have implemented a lemma for each \ac{KDF} application that proves the result's secrecy label.

Moreover, the invariant for protocol-specific events has to be strong enough to prove injective agreement with \acs{AKC} resilience.
Our reusable verification library enables strengthening the proven authentication property from non-injective to injective by adding the uniqueness witness for each protocol-specific event.
This allowed us to focus on finding a suitable invariant for  non-injective agreement with \acs{AKC} resilience first, and then 
strengthen the authentication property, which required less than 40~additional \acp{LOS}.

After completing the proofs for sequential code, we re-introduced concurrency to the initiator's transport phase (recall \secref{wireguard}), which entailed only minimal proof changes and was done in an afternoon. This demonstrates that our separation-logic-based methodology enables security proofs that are robust \wrt nontrivial code changes.

%% file: 7_trust_assumptions.tex
\section{Trust Assumptions and Soundness}
\label{sec:trust-assumptions}

Our methodology allows us to prove strong security properties for implementations of security protocols. Like with all verification techniques, these proofs rely on several assumptions about the implementation and the execution environment.

We rely on the soundness of the used program verifier. Since our methodology is compatible with standard separation logic verifiers, we can mitigate this assumption by using a mature tool.

As is standard for symbolic cryptography, we assume perfect cryptographic operations (\eg absence of hash collisions, or that ciphertexts do not leak any information). We also do not verify that the implementations of the cryptographic primitives are functionally correct; while this is orthogonal to our work, our methodology could be combined with verified libraries like EverCrypt~\cite{DBLP:conf/sp/ProtzenkoPFHPBB20}.

Furthermore, we assume that all \emph{output} operations, \ie sending messages, are reflected on the global trace by corresponding events, which is the case when using the I/O operations provided by our verification library.
However, if an implementation uses, \eg inline assembly or third-party libraries to send messages to the network, the global trace has to reflect these messages nonetheless. 
Omitting any other event does not affect soundness, only completeness.

Lastly, we assume that the protocol terms corresponding to the byte arrays in a participant's \emph{initial} state, and those obtained from operations outside of our library (\eg read from a config file), are readable at least by that participant according to the terms' secrecy labels (recall \secref{secrecy}).
Otherwise, it would not be sufficient to model corruption of a participant by adding the class of terms readable by that participant to the attacker knowledge; the attacker could learn even more.
For all terms a participant can obtain by interacting with our verification library (\eg receiving messages, generating nonces, applying encryption), we prove in our library (via corresponding lemmas) that a participant can read these terms (and thus the terms leak as expected to the attacker in case of corruption).

We sketch soundness of our methodology in \ifthenelse{\boolean{soundness_proof}}{\appref{soundness-proof-sketch}}{\cite{full-version}} by showing that the global trace reflects all relevant protocol steps and, thus, any security property proved for the trace indeed holds for the protocol implementation.
For this purpose, we define a core programming language covering all protocol-relevant operations (\eg network I/O, cryptographic primitives), and those relevant for modeling an attacker (\eg corruption).
The language's operational semantics supports thread-local state and
explicitly maintains a shared global trace.
The thread-local state models the state of each participant and is manipulated via assignments in the participants' implementations. In contrast, the global trace is maintained automatically by our semantics and extended whenever a relevant protocol operation is executed.
We then define a Hoare logic to enable modularly verifying each participant implementation. The logic natively supports our methodology's local snapshots and the global trace. We prove that this logic is sound \wrt the operational semantics using a standard rule induction. Thereby, we obtain the guarantee that locally-verified participants, if composed with the attacker to a concurrent system, maintain the global trace invariant in all possible interleavings.
As one would expect, obtaining this global guarantee turned out to be the most challenging step in the soundness proof. Formally connecting our dedicated Hoare logic to a standard separation logic is straightforward, based on the encoding discussed throughout the paper (using the heap to store the thread-local state, a ghost lock to synchronize access, and a lock invariant to constrain the trace and all local snapshots).

%% file: 8_related_work.tex
\section{Related Work}
\label{sec:related_work}

  Much prior work on the verification of cryptographic protocols exists, and surveys~\cite{BarbosaBBBCLP21,AvallePS14,Blanchet12} provide an extensive overview of the field.
  We focus on \emph{modular verification of symbolic security properties}, and discuss the most closely related work first: techniques for verifying security of \emph{realistic protocol implementations}.

\citeauthors{Dupressoir}{DupressoirGJN11} use VCC~\cite{CohenDHLMSST09} to verify memory safety, non-injective agreement, and (via an external argument in Coq) weak secrecy, of two protocols implemented in C:
RPC and Otway-Rees.
To our knowledge, they are the first to encode a global protocol trace (``log'') as a concurrent data structure.
We generalize this idea to separation logic to make it much more widely applicable,
because their encoding relies on C's \lstinline@volatile@ fields and a VCC-specific program logic, neither of which are (widely) available in other languages and verifiers.
Moreover, since their logic (unlike separation logic) does not provide linear resources, proving injective agreement would require a nontrivial extension of their work. 
Their set-based trace encoding prevents proving, \eg forward secrecy (which we do); they account for principal corruption, but not session corruption (we account for both).
\citeauthors{Polikarpova}{PolikarpovaM12} extend this work by incorporating stepwise refinement to formally connect a model to an existing implementation, all encoded in VCC\@. This refinement decomposes the verification into smaller steps, but incurs additional overhead. 
Moreover, they remove the need for external arguments 
when proving weak secrecy. 
They verify the latter, and a variant of authentication, for a small stateful subset of TPM~2.0.

\looseness=-1
\citeauthors{Vanspauwen}{Vanspauwen015}, like us, use a separation-logic-based verifier (VeriFast~\cite{JacobsSPVPP11}), but they do not model a global trace. Consequently, properties that are commonly expressed over a trace potentially need to be assembled from individual assertions. They propose an extended symbolic model 
that strengthens attackers by permitting byte-wise manipulations, such as splitting and reconcatenating byte sequences, in addition to the usual symbolic manipulations. 
Our attacker operates on terms (standard for symbolic cryptography) but we could adapt their extension.
They specify PolarSSL's API using this extended model, and then verify secrecy and non-injective agreement of an NSL-implementation (and a few less complex protocols).
Unlike us, they do not consider session corruption.

\looseness=-1
\citeauthors{Arquint}{ArquintWLSSWBM23} suggest a two-step approach: 
First, a protocol \emph{model} is verified via Tamarin~\cite{SchmidtMCB12}. If successful, a separation logic predicate (one per participant) with I/O specifications~\cite{Penninckx0P15} is generated, specifying which I/O operations preserve the security properties of the model.
Second, existing implementations of the protocol can be verified against these predicates; if successful, the implementation is guaranteed to satisfy the model's properties.
This two-step workflow achieves tool reuse---Tamarin, and suitable separation logic verifiers---but requires expertise in two different fields of formal reasoning, and the existence of a Tamarin protocol model.
Moreover, limitations of Tamarin (\eg difficulties when verifying protocols with loops), and of the I/O specifications (unclear how to generate specifications suitable for a concurrent implementation) may prevent verifying corresponding implementations. Similar limitations apply to Sprenger et~al.'s work~\cite{DBLP:journals/pacmpl/0001KEW0CB20}, which connects protocol models verified in Isabelle/HOL~\cite{dblp:books/sp/nipkowpw02} via I/O specifications to separation logic verifiers.

  \citeauthors{Bhargavan}{BhargavanBDHKSW21} suggest DY*: a framework for verifying protocols implemented in F*~\cite{SwamyHKRDFBFSKZ16}, a functional language that enables type-system-based proofs, \eg using monadic effects and refinement types. DY* introduces the idea of a parametric library for reducing the per-protocol proof effort; an idea we adopted. DY*'s type system is tailored to F*, whereas our methodology supports a wide range of languages and tools. Moreover, by building on separation logic, we are able to prove stronger properties, in particular, injective agreement. Our methodology can be applied directly to existing implementations, as we demonstrate in the WireGuard case study. In contrast, DY* supports code generation, but additionally requires a hand-written (and partly protocol-specific) runtime wrapper~\cite{BhargavanB0HKSW21}.
  Included in DY*'s case study is the first automated verification of Signal~\cite{MarlinspikeP16} that proves forward and post-compromise security over an unbounded number of protocol messages. Our main case study is WireGuard, for which we prove, also for an unbounded number of messages, forward secrecy and injective agreement with \acs{AKC} resilience.
  Soundness of DY*'s global protocol trace depends on a specific coding discipline (one method per protocol step) that is not automatically enforced. If missed, the attacker is accidentally restricted, and security properties can be proven incorrectly.

An earlier line of work (\eg~\cite{BengtsonBFGM08,BhargavanFG10,BhargavanFG10b}) verifies security of functional programs written in F\# using the F7 type checker~\cite{BengtsonBFGM08}, but does not integrate equational theories, and has limited support for mutable state. Moreover, this work does not model the global protocol traces and, thus, states security properties only implicitly.

\citeauthors{K{\"{u}}sters}{KustersTG12} share our goal of reusing existing program analyzers 
and suggest an approach that enables non-interference checkers to establish computational indistinguishability results for sequential programs. To account for closed-system assumptions (typically made by such checkers) in the presence of an attacker-controlled environment, they restrict interaction with the latter to static, exception-free methods, and primitive (\ie value) types. How to extend their approach to trace-based properties and concurrent programs remains unclear.

Several security property verifiers exist that (unlike us) do not reuse existing program analyzers, \eg Csur~\cite{Goubault-LarrecqP05} and ASPIER~\cite{ChakiD09} (for C), and JavaSec~\cite{Jurjens06} (for Java). However, to reduce development costs, such domain-specific tools typically only implement semantics of a restricted language subset and, \eg assume crucial properties such as memory safety (which may render implementations insecure, \eg due to buffer overflows).

Prior work \cite{DonenfeldM18,DowlingP18,KobeissiNB19,LippBB19,GirolHSJCB20,HoPBB22} on verifying properties of WireGuard (our main case study) is concerned with verifying models of the protocol, not existing implementations.

Finally, a large body of work is concerned with mechanizing the verification of computational (rather than symbolic) properties; see aforementioned surveys for details.
This line of work establishes stronger guarantees by making weaker, more realistic cryptographic assumptions. For instance, Owl~\cite{Gancher23} allows one to verify computational security of protocols written in a dedicated language. Like in our work, their proofs are automated and compositional. However, due to probabilistic reasoning, verifying computational security is generally more challenging than reasoning about symbolic terms, and we are not aware of tools for modularly verifying computational security properties of existing implementations.
Recently, the first separation logics for probabilistic reasoning have been proposed \cite{TassarottiH19,BatzKKMN19,BartheHL20}, but we are not aware of automated verifiers for such logics. 

%% file: 9_conclusions.tex
\section{Conclusions}
\label{sec:conclusions}

\looseness=-1
We presented a methodology for the modular verification of security protocol implementations. It enables proving strong security properties for realistic protocol implementations in the presence of a network-controlling attacker. By employing separation logic, we support efficient implementations using heap data structures, side effects, concurrency, etc. Encapsulating the global trace in a concurrent ghost data structure and our use of invariants over local snapshots allow our methodology to support arbitrary code structures and data representations, which is crucial for targeting existing implementations. Separation logic also allows us to specify resources in the trace invariant to express uniqueness of protocol-specific events, which is key to modularly proving injective agreement.

We have instantiated our methodology for Go and C and two corresponding verifiers. Our case studies on NSL, signed DH, and WireGuard demonstrate that our methodology handles existing and interoperable implementations of protocols with strong security properties, such as forward secrecy and injective agreement.

For future work, we plan to integrate our methodology with formally-verified cryptographic libraries to further reduce our trust assumptions. It would also be interesting to advance towards the computational model of cryptography by combining our work with probabilistic separation logic. 

\begin{acks}
	We thank the Werner Siemens-Stiftung (WSS) for their generous support of this project.
	We are grateful to Ralf Sasse for the helpful discussions and feedback on an earlier draft of this paper;
	Thibault Dardinier, for his suggestions for the soundness proof;
	Hugo Queinnec, for his critical questions;
	and the anonymous reviewers for their helpful feedback that helped us sharpen our contributions.
\end{acks}

%% file: AD_appendix.tex
\section{Soundness Proof Sketch}
\label{sec:soundness-proof-sketch}

Intuitively, we argue soundness of our methodology by showing that, given a distributed system of verified protocol implementations and an arbitrary attacker, the systems' set of possible executions is a subset of the executions permitted by the verification trace invariant, which in turn is a subset of the executions that satisfy the desired security properties.
To achieve this, we define a minimal but concurrent programming language with primitives for security-relevant operations such as sending messages or creating nonces, and a corresponding operational semantics (\secref{operational-semantics}) that reflects these operations on a global (i.e. system-wide shared) trace.
We then define an axiomatic semantics (\secref{program-logic}) parameterized with a trace invariant that we prove sound \wrt the operational semantics. \Ie we show that the axiomatic proof rules enforce the trace invariant.
Since the global trace maintained by the operational semantics reflects all relevant protocol steps, and because our axiomatic semantics is proven sound, we can conclude that the aforementioned trace inclusion holds (\secref{trace-inclusion}).
In each subsection, we additionally relate the semantics defined for the proof sketch with the verification performed by an off-the-shelf separation-logic verifier (such as Gobra) against our reusable verification library.

\subsection{Language and Operational Semantics}
\label{sec:operational-semantics}

On a high-level, we consider a distributed system consisting of multiple components: either instances of a protocol implementation, i.e. participants, or the attacker. Our programming language does not support user-defined shared variables or a heap, and each participant executes its commands in its own local state. However, security-relevant commands additionally mutate the global trace to reflect the performed operation.

Consequently, our system's configurations comprise a local configuration per component, and the global trace~$\tau$. A local configuration for a protocol participant~$i$ is characterized by its local command~$C_i$ and local state~$\sigma_i$. The local configuration for the attacker is similar, but additionally contains a knowledge set~$k_a$ that stores all symbolic terms that are known to the attacker.

\begin{definition}\textbf{Local program states.}
    Local program states, ranged over by $\sigma$, are total functions from local variables (in the set $\vars$) to values (in the set $\vals$).
    $$\states \triangleq \vars \rightarrow \vals$$
\end{definition}

We define our programming language such that it directly works with symbolic terms instead of bytes, which avoids having to complicate the semantics to reflect the orthogonal issue of mapping between the bytes and terms.

\begin{definition}\textbf{System configurations.}\label{def:system_configuration}
	A configuration of our distributed system has the shape
	\begin{align*}
		\systemstatefive{\participantcomponent{C_1}{\sigma_1}}{\cdots}{\participantcomponent{C_n}{\sigma_n}}{\participantcomponent{C_a}{\sigma_a}}{k_a}{\tau}
	\end{align*}
	where $\participantcomponent{C_i}{\sigma_i}$ denotes the local command and local state of participant~$i$, $\participantcomponent{C_a}{\sigma_a}$ denotes the local command and local state of the attacker~$a$, $k_a$ is the attacker's knowledge set and $\tau$ denotes the system's global trace.
\end{definition}

Observe that the attacker's knowledge set~$k_a$ is not part of the attacker's local configuration $\participantstate{C_a}{\sigma_a}$, even though only commands executed by the attacker possibly modify~$k_a$. By using the same shape for local configurations of participants and the attacker, both can apply the same operational semantics rules, \eg for sequential composition.

\begin{definition}\textbf{Programming language.}\label{def:language}
We consider the following programming language, where $C$ ranges over commands, $x$ and $\vec{x}$ over variables and lists of variables in the set $\vars$, respectively, and $e$ over expressions (modeled as total functions from $\states$ to $\vals$):
\begin{align*}
C \triangleq
&\;
\cskip \mid
C \cseq C \mid
\cif{e}{C}{C} \mid
\cwhile{e}{C} \mid
\\
&\;
\cassign{x}{e} \mid
\csend{e} \mid
\crecv{x} \mid
\cnonce{x} \mid
\\
&\;
\chash{x}{e} \mid
\cpk{x}{e} \mid
\cenc{x}{e}{e} \mid
\cdec{x}{x}{e}{e} \mid
\\
&\;
\cdrop{e} \mid
\clearn{e} \mid
\cchoose{x} \mid
\ccorrupt{e} \mid
\\
&\;
\cfork{\vec{x}}{C}
\end{align*}
\end{definition}

Besides standard commands, such as sequential composition and assignment, the programming language provides several commands essential for protocol implementations: for sending and receiving a network message, for generating a nonce, hashing a term, generating a public key corresponding to a given secret key ($\mathit{pk}$), and encrypting and decrypting a term with a key.
Additionally, the programming language provides commands only available to the attacker: dropping a message from the network, adding the value of a local variable to the attacker knowledge ($\mathit{learn}$), non-deterministically obtaining a term from the attacker knowledge ($\mathit{choose}$), and corrupting the state of specific participant (each participant has a unique id/index).

Finally, fork starts a new thread executing the provided command, which corresponds to spawning a new participant or the attacker. The new thread operates on its own local state, which initially maps the variables in $\vec{x}$ to the same values as the state in which the fork command is executed. This command is used to bootstrap the distributed system, as discussed in \secref{trace-inclusion}.

The expression language comprises symbolic terms for booleans and integers, and the usual operations thereon. We assume well-typed programs, \eg that if-conditions are of type boolean.

\begin{definition}\textbf{Operational semantics}\label{def:operational_semantics}
	\Figref{semantics} defines the small-step operational semantics for our programming language.
\end{definition}

\begin{figure*}
\begin{center}
\[
\begin{array}{c}
\Inf[\rulename{Local}]{\smallsteppart{C_i}{\participantstate{\sigma_i}{\tau}}{C'_i}{\participantstate{\sigma'_i}{\tau'}}}{\smallstepsystem{\systemstatefour{\cdots}{\participantcomponent{C_i}{\sigma_i}}{\cdots}{k_a}{\tau}}{\systemstatefour{\cdots}{\participantcomponent{C'_i}{\sigma'_i}}{\cdots}{k_a}{\tau'}}}

\hspace{5mm}

\Inf[\rulename{Attacker}]{\smallsteppart{C_a}{\attackerstate{\sigma_a}{k_a}{\tau}}{C'_a}{\attackerstate{\sigma'_a}{k'_a}{\tau'}}}{\smallstepsystem{\systemstatethree{\cdots}{\participantcomponent{C_a}{\sigma_a}}{k_a}{\tau}}{\systemstatethree{\cdots}{\participantcomponent{C'_a}{\sigma'_a}}{k'_a}{\tau'}}}

\\[2em]

\Inf[\rulename{Skip}]{\smallsteppart{\cskip}{\participantstate{\sigma}{\tau}}{\cempty}{\participantstate{\sigma}{\tau}}}

\hspace{5mm}

\Inf[\rulename{Seq1}]{\smallsteppart{C_1}{\participantstate{\sigma}{\tau}}{\cempty}{\participantstate{\sigma'}{\tau'}}}{\smallsteppart{C_1 \cseq C_2}{\participantstate{\sigma}{\tau}}{C_2}{\participantstate{\sigma'}{\tau'}}}

\hspace{5mm}

\Inf[\rulename{Seq2}]{\smallsteppart{C_1}{\participantstate{\sigma}{\tau}}{C'_1}{\participantstate{\sigma'}{\tau'}}}{\smallsteppart{C_1 \cseq C_2}{\participantstate{\sigma}{\tau}}{C'_1 \cseq C_2}{\participantstate{\sigma'}{\tau'}}}

\\[2em]

\Inf[\rulename{If1}][e(\sigma_i) = \ttrue]{\smallsteppart{\cif{e}{C_1}{C_2}}{\participantstate{\sigma}{\tau}}{C_1}{\participantstate{\sigma}{\tau}}}

\hspace{5mm}

\Inf[\rulename{If2}][e(\sigma_i) \neq \ttrue]{\smallsteppart{\cif{e}{C_1}{C_2}}{\participantstate{\sigma}{\tau}}{C_2}{\participantstate{\sigma}{\tau}}}

\\[2em]

\Inf[\rulename{While}]{\smallsteppart{\cwhile{e}{C}}{\participantstate{\sigma}{\tau}}{\cif{e}{C \cseq \cwhile{e}{C}}{\cskip}}{\participantstate{\sigma}{\tau}}}

\\[2em]

\Inf[\rulename{Assign}]{\smallsteppart{\cassign{x}{e}}{\participantstate{\sigma}{\tau}}{\cempty}{\participantstate{\sigma[x \mapsto e(\sigma)]}{\tau}}}

\\[2em]

\Inf[\rulename{Send}]{\smallsteppart{\csend{e}}{\participantstate{\sigma}{\tau}}{\cempty}{\participantstate{\sigma[\snapshot \mapsto \tau + \esend{e(\sigma)}]}{\tau + \esend{e(\sigma)}}}}

\\[2em]

\Inf[\rulename{Recv}][v \in \receivable{\tau}]{\smallsteppart{\crecv{x}}{\participantstate{\sigma}{\tau}}{\cempty}{\participantstate{\sigma[x \mapsto v]}{\tau}}}

\\[2em]

\Inf[\rulename{NonceGen}][\fresh{v}{\tau}]{\smallsteppart{\cnonce{x}}{\participantstate{\sigma}{\tau}}{\cempty}{\participantstate{\sigma[x \mapsto v, \snapshot \mapsto \tau + \enonce{v}]}{\tau + \enonce{v}}}}

\\[2em]

\Inf[\rulename{Hash}]{\smallsteppart{\chash{x}{e}}{\participantstate{\sigma}{\tau}}{\cempty}{\participantstate{\sigma[x \mapsto \thash{e(\sigma)}]}{\tau}}}

\hspace{5mm}

\Inf[\rulename{Pk}]{\smallsteppart{\cpk{x}{e}}{\participantstate{\sigma}{\tau}}{\cempty}{\participantstate{\sigma[x \mapsto \tpk{e(\sigma)}]}{\tau}}}

\\[2em]

\Inf[\rulename{Enc}]{\smallsteppart{\cenc{x}{e_1}{e_2}}{\participantstate{\sigma}{\tau}}{\cempty}{\participantstate{\sigma[x \mapsto \tenc{e_1(\sigma)}{e_2(\sigma)}]}{\tau}}}

\\[2em]

\Inf[\rulename{DecSucc}][\exists v \ldotp e_2(\sigma) = \tenc{\tpk{e_1(\sigma)}}{v}]{\smallsteppart{\cdec{x}{\vok}{e_1}{e_2}}{\participantstate{\sigma}{\tau}}{\cempty}{\participantstate{\sigma[x \mapsto v, \vok \mapsto \ttrue]}{\tau}}}

\\[2em]

\Inf[\rulename{DecFail}][\forall v \ldotp e_2(\sigma) \neq \tenc{\tpk{e_1(\sigma)}}{v}]{\smallsteppart{\cdec{x}{\vok}{e_1}{e_2}}{\participantstate{\sigma}{\tau}}{\cempty}{\participantstate{\sigma[\vok \mapsto \tfalse]}{\tau}}}

\\[2em]

\Inf[\rulename{Drop}]{\smallsteppart{\cdrop{e}}{\attackerstate{\sigma}{k}{\tau}}{\cempty}{\attackerstate{\sigma[\snapshot \mapsto \tau + \edrop{e(\sigma)}]}{k}{\tau + \edrop{e(\sigma)}}}}

\\[2em]

\Inf[\rulename{Learn}]{\smallsteppart{\clearn{e}}{\attackerstate{\sigma}{k}{\tau}}{\cempty}{\attackerstate{\sigma}{k \cup \{ e(\sigma) \}}{\tau}}}

\hspace{5mm}

\Inf[\rulename{Choose}][v \in k]{\smallsteppart{\cchoose{x}}{\attackerstate{\sigma}{k}{\tau}}{\cempty}{\attackerstate{\sigma[x \mapsto v]}{k}{\tau}}}

\\[2em]

\Inf[\rulename{Corrupt}]{\smallstepsystem{\systemstatefive{\cdots}{\participantcomponent{C_i}{\sigma_i}}{\cdots}{\participantcomponent{\ccorrupt{i} \cseq C'_a}{\sigma_a}}{k_a}{\tau}}{\systemstatefive{\cdots}{\participantcomponent{C_i}{\sigma_i}}{\cdots}{\participantcomponent{C'_a}{\sigma_a}}{k_a \cup \values{\sigma_i}}{\tau + \ecorrupt{i}{\values{\sigma_i}}}}}

\\[2em]

\Inf[\rulename{Fork}]{\smallstepsystem{\systemstatefour{\cdots}{\participantcomponent{\cfork{x_1, \cdots, x_n}{C} \cseq C'}{\sigma_i}}{\cdots}{k_a}{\tau}}{\systemstatefive{\cdots}{\participantcomponent{C'}{\sigma_i}}{\cdots}{\participantcomponent{C}{[x_1 \mapsto \sigma_i(x_1), \cdots, x_n \mapsto \sigma_i(x_n), \snapshot \mapsto \sigma_i(\snapshot)]}}{k_a}{\tau}}}

\end{array}
\]
\end{center}
\caption{Small-step semantics. Since expressions are functions from states to values, $e(\sigma)$ denotes the evaluation of expression $e$ in state $\sigma$. $\sigma[x_1 \mapsto v_1, \cdots, x_n \mapsto v_n]$ denotes state update: a state that, for all $i$, $1 \leq i \leq n$, yields $v_i$ for $x_i$, and the value in $\sigma$ for all other variables. Appending to a trace is denoted by $+$, \eg $\tau + \enonce{v}$. $\config{\cempty}{\participantstate{\sigma}{\tau}}$ denotes a terminal state.}
\label{fig:semantics}
\end{figure*}

The rules for standard commands such as sequential composition and conditionals, are as expected, and we will thus only discuss non-standard aspects of our programming language.

\paragraph{Global trace}
Recall from \secref{global-trace-encoding} that in our verification methodology (as implemented in Gobra), we use a concurrent ghost data structure with ghost locks to manage the global trace. In our operational semantics, we instead represent the trace as the dedicated element $\tau$ in the system's state. Irregardless of the technical implementation we must ensure three crucial properties:
(1) Each operation may only append a single trace events. In our methodology, this is checked via a suitable proof obligation upon lock release; in our operational semantics, each rule adds at most one event.
(2) To ensure monotonicity, the trace may only grow. Checked upon lock release in our methodology; in our operational semantics, no rule shortens the trace.
(3) Each single operation must preserve the trace invariant. Checked upon lock release in our methodology; in our operational semantics, this is part of the soundness theorem (\cf \thmref{soundness}).

\paragraph{Local snapshots}
Recall from \cf~\secref{local-snapshots} that each participant has a trace snapshot, which enables participants to keep local invariants of trace prefixes. To enable corresponding assertions in our program logic (\secref{program-logic}), our operational semantics provide a local variable~\snapshot that is treated special in two ways: local states $\sigma$ map \snapshot to a sequence of trace events (not to a value in $\vals$), and program commands may not use (in particular, modify) \snapshot (a straightforward syntactical constraint).

\paragraph{Projecting system configurations}
The \rulename{Local} and \rulename{Attacker}~rule project project a system configuration down to a participant- and attacker-local configuration, respectively. Besides \rulename{Corrupt} and \rulename{Fork}, all other rules then operate on either a participant- or attacker- local configuration, depending on whether a command can be executed by participants and the attacker, or only by the attacker.

\paragraph{Network messages}
All operations modifying the network state, \ie sending and dropping a message, are recorded on the global trace~$\tau$, and we can thus compute the set of receivable messages as follows:

\begin{definition}\textbf{Messages on the network.}\label{def:receivable}
	\begin{align*}
	\receivable{\tau} \triangleq &\;
  	\{m \mid \forall m \ldotp \event{Send(m)} \in \tau \land \event{Drop(m)} \notin \tau\}
	\end{align*}
\end{definition}

Consequently, function $\receivable{\tau}$ occurs in rule \rulename{Recv}'s side-condition to constrain the set of messages to receive from Without loss of generality, this side-condition implies that we consider only non-blocking traces, \ie where $\mathit{recv}()$ is invoked when $\receivable{\tau}$ is non-empty.

\paragraph{Nonce freshness}
\rulename{NonceGen}~rule's side condition captures our perfect cryptography assumption that generated nonces are always fresh.

\begin{definition}\textbf{Freshness of nonces.}\label{def:fresh}
	Since all previously generated nonces have been recorded on the trace~$\tau$, we can define freshness of a nonce~$v$ on the global trace~$\tau$ as follows:
	\begin{align*}
	\fresh{v}{\tau} \triangleq &\;
  	v \notin \{n \mid \forall n, l \ldotp \event{Nonce(n, l)} \in \tau\}
	\end{align*}
\end{definition}

\paragraph{Corruption}
The \rulename{Corrupt}~rule expresses that the attacker knowledge is extended by all terms in the state~$\sigma_i$ of the corrupted participant~$i$. The attacker can make use of these newly learnt terms by executing $\cchoose{x}$ that non-deterministically picks a term in the attacker knowledge and assigns it to the local variable~$x$.

Note that the \rulename{Corrupt}~rule requires command $\ccorrupt{}$ to be followed by another command, skip. Baking in sequential composition avoids the need for further sequential rules that only differ in the kind of configuration (system vs. local) they operate on. Rule \rulename{Fork} is defined analogously.

\paragraph{Forking}
Rule \rulename{Fork} extends the system configuration by another local configuration with the forked command to execute and the new thread's initial state. This new state maps the variables~$x_1$ to $x_n$ and \snapshot to the same value as $\sigma_i$, \ie the state in which the fork command is executed, which enables, \eg the sharing of public keys.

\subsection{Program Logic}
\label{sec:program-logic}

We now present a program logic that enables local reasoning about each participant, while guaranteeing that the trace invariant is maintained even when composing arbitrarily many verified participants and the attacker to a distributed system. We first present several auxiliary definitions and lemmas, and then the logic's proof rules.

\begin{definition}\textbf{Trace prefix.}\label{def:prefix}
	We define the following predicate over two~traces expressing that $\tau_1$ is a prefix of $\tau_2$
	\begin{align*}
		\prefix{\tau_1}{\tau_2} \triangleq \; \exists p \ldotp \tau_1 + p = \tau_2
	\end{align*}
	where $p$ is a possibly empty sequence of trace events.
\end{definition}

\begin{lemma}\textbf{Prefix reflexivity.}\label{lem:prefix-reflexivity}
	\begin{align*}
		& \forall \tau \ldotp \prefix{\tau}{\tau}
	\end{align*}
\end{lemma}
\begin{proof}
	Pick $p$ to be the empty sequence in \defref{prefix}.
\end{proof}

\begin{lemma}\textbf{Prefix transitivity.}\label{lem:prefix-transitivity}
	\begin{align*}
		& \forall \tau_1, \tau_2, \tau_3 \ldotp \prefix{\tau_1}{\tau_2} \land \prefix{\tau_2}{\tau_3} \implies \prefix{\tau_1}{\tau_3}
	\end{align*}
\end{lemma}
\begin{proof}
	\begin{align*}
		& \prefix{\tau_1}{\tau_2} \land \prefix{\tau_2}{\tau_3} \\
		& \iffdef \exists p_1, p_2 \ldotp \tau_1 + p_1 = \tau_2 \land \tau_2 + p_2 = \tau_3 \\
		& \implies \exists p_1, p_2 \ldotp \tau_1 + p_1 + p_2 = \tau_3 \\
		& \iffdef \prefix{\tau_1}{\tau_3}
	\end{align*}
	where we pick in the last step $p$ in \defref{prefix} to be $p_1 + p_2$.
\end{proof}

Inspired by Vafeiadis~\cite{DBLP:journals/entcs/Vafeiadis11}, we express the semantics of judgements in our logic in terms of configuration safety, which we define next.
Intuitively, \safe{n}{i}{C}{\sigma}{Q}{\tau} expresses that it is safe to execute command~$C$, as the $i$th component of the distributed system, and for $n$~execution steps starting in a state~$\sigma$; and if the command is fully executed, the predicate~$Q$ holds in the resulting final state. Furthermore, if new threads have been forked as part of executing~$C$ then it is safe to execute these forked components, too.
Since we are ultimately interested in the effects on the global trace~$\tau$, configuration safety includes trace invariant~\traceinvariant maintenance.
A judgement~\hoare{P}{C}{Q} then expresses that it is safe to execute the command~$C$ for an arbitrary number of execution steps and from any initial state satisfying the predicate~$P$.

\begin{definition}\textbf{Configuration safety.}\label{def:config-safety}
	\begin{align*}
		& \safe{0}{i}{C}{\sigma}{Q}{\tau} \; \text{holds always.} \\
		& \safe{n + 1}{i}{C}{\sigma}{Q}{\tau} \; \text{holds if and only if} \\
		& \; \; \text{(i)} \; C = \cempty \implies Q(\sigma) \; \text{and} \\
		& \; \; \text{(ii)} \; \forall \vec{C}, \vec{C'}, \vec{\sigma}, \vec{\sigma'}, k_a, k'_a, \tau' \ldotp i \leq |\vec{C}| = |\vec{\sigma}| \leq |\vec{C'}| = |\vec{\sigma'}| \land {} \\
		& \; \phantom{\forall \vec{C}, \vec{C'}, \vec{\sigma}, \vec{\sigma'}, k_a, k'_a, \tau' \ldotp} \vec{C_i} = C \land \vec{C'_i} \neq C \land \vec{\sigma_i} = \sigma \land {} \\
		& \; \phantom{\forall \vec{C}, \vec{C'}, \vec{\sigma}, \vec{\sigma'}, k_a, k'_a, \tau' \ldotp} \traceinvariant(\tau) \land \prefix{\snapshot(\sigma)}{\tau} \land {} \\
		& \; \phantom{\forall \vec{C}, \vec{C'}, \vec{\sigma}, \vec{\sigma'}, k_a, k'_a, \tau' \ldotp} \smallstepsystem{\systemstatetwo{\overrightarrow{\participantstate{C}{\sigma}}}{k_a}{\tau}}{\systemstatetwo{\overrightarrow{\participantstate{C'}{\sigma'}}}{k'_a}{\tau'}} \\
		& \phantom{\; \; \text{(ii)} \;} \implies \traceinvariant(\tau') \land \prefix{\tau}{\tau'} \land \prefix{\snapshot(\vec{\sigma'_i})}{\tau'} \land {} \\
		& \phantom{\; \; \text{(ii)} \; \implies} k_a \subseteq k'_a \land \safe{n}{i}{\vec{C'_i}}{\vec{\sigma'_i}}{Q}{\tau'} \land {} \\
		& \phantom{\; \; \text{(ii)} \quad} \left(\bigwedge_{|\vec{C}| < j \leq |\vec{C'}|} \safe{n}{j}{\vec{C'_j}}{\vec{\sigma'_j}}{\ttrue}{\tau'} \land \prefix{\snapshot(\vec{\sigma'_j})}{\tau'}\right)
	\end{align*}
	where $|\vec{V}|$ and $\vec{V_i}$ denote the length and element at index~$i$ of a vector~$V$, resp., and $\overrightarrow{\participantstate{C}{\sigma}}$ is syntactic sugar for $\participantstate{\vec{C_1}}{\vec{\sigma_1}} \cdots \participantstate{\vec{C}_{|\vec{C}|}}{\vec{\sigma}_{|\vec{\sigma}|}}$.
\end{definition}

\begin{definition}\textbf{Validity.}\label{def:validity}
	\begin{align*}
		\hoare{P}{C}{Q} \triangleq \; \forall n, i, \sigma, \tau \ldotp P(\sigma) \implies \safe{n}{i}{C}{\sigma}{Q}{\tau}
	\end{align*}
\end{definition}

Executing zero steps is vacuously safe. Executing $n + 1$ steps is safe $(i)$~if the command is already fully executed and the predicate~$Q$ satisfied; and otherwise $(ii)$~if there is a transition to $\vec{C'_i}$ that maintains the trace invariant~\traceinvariant, the necessary monotonicity properties (on snapshot, trace, and the attacker's knowledge set), and allows continued safe execution of all components (\ie of commands~$\vec{C'}$) in the system, including newly forked ones (the last, iterated conjunct in the definition).

\begin{figure*}[t!]
\begin{center}
    \[
    \begin{array}{c}

    \Inf[\rulename{Skip}]
    	{\phantom{\shoare{}{}{}}} 
    	{\shoare
    		{P}
    		{skip}
    		{P}}

    \hspace{5mm}
    
	\Inf[\rulename{Seq}]
    	{\shoare{P}{C_1}{R}}
    	{\shoare{R}{C_2}{Q}}
    	{\shoare{P}{C_1 \cseq C_2}{Q}}

    \hspace{5mm}
    
	\Inf[\rulename{Cons}]
		{P \models P'}
		{Q' \models Q}
		{\shoare{P'}{C}{Q'}}
		{\shoare{P}{C}{Q}}
		
	\\[2em]
    
    \Inf[\rulename{If}]
    	{\shoare{e \land P}{C_1}{Q}}
    	{\shoare{\neg e \land P}{C_2}{Q}}
    	{\shoare{P}{\cif{e}{C_1}{C_2}}{Q}}

    \hspace{5mm}

    \Inf[\rulename{While}]
    	{\shoare{e \land P}{C}{P}}
    	{\shoare{P}{\cwhile{e}{C}}{\neg e \land P}}

    \hspace{5mm}

    \Inf[\rulename{Assign}]
    	{\phantom{\shoare{}{}{}}} 
    	{\shoare
    		{P[e/x]}
    		{\cassign{x}{e}}
    		{P}}

    \\[2em]

    \Inf[\rulename{Send}]
    	{\shoare
    		{\ext{\esend{e}}{\snapshot} \land \forall p \ldotp P[\snapshot + p + \esend{e} / \snapshot]}
    		{\csend{e}}
    		{P}}

    \hspace{5mm}
    
    \Inf[\rulename{Recv}]
    	{\shoare
    		{\forall x \ldotp P}
    		{\crecv{x}}
    		{P}}

    \\[2em]
    
    \Inf[\rulename{NonceGen}]
    	{\shoare
    		{\ext{\enonce{x}}{\snapshot} \land \forall p, x \ldotp P[\snapshot + p + \enonce{x} / \snapshot]}
    		{\cnonce{x}}
    		{P}}
    \\[2em]
    
    \Inf[\rulename{Hash}]{
    	\shoare
    		{P[\thash{e}/x]}
    		{\chash{x}{e}}
    		{P}}

    \hspace{5mm}
    
    \Inf[\rulename{Pk}]{
    	\shoare
    		{P[\tpk{e}/x]}
    		{\cpk{x}{e}}
    		{P}}
	
	\hspace{5mm}
	
	\Inf[\rulename{Enc}]{
    	\shoare
    		{P[\tenc{e_1}{e_2}/x]}
    		{\cenc{x}{e_1}{e_2}}
    		{P}}

	\\[2em]
	
	\Inf[\rulename{Dec}]{
    	\shoare
    		{\forall x \ldotp P[\ttrue / \vok][e_2 / \tenc{\tpk{e_1}}{x}] \land P[\tfalse / \vok]}
    		{\cdec{x}{\vok}{e_1}{e_2}}
    		{P}}

    \\[2em]
	
	\Inf[\rulename{Drop}]
    	{\shoare
    		{\ext{\edrop{e}}{\snapshot} \land \forall p \ldotp P[\snapshot + p + \edrop{e} / \snapshot]}
    		{\cdrop{e}}
    		{P}}
	
	\\[2em]
	
	\Inf[\rulename{Learn}]
    	{\shoare
    		{P}
    		{\clearn{e}}
    		{P}}
	
	\hspace{5mm}
	
	\Inf[\rulename{Choose}]
    	{\shoare
    		{\forall x \ldotp P}
    		{\cchoose{x}}
    		{P}}
	
	\\[2em]
	
    \Inf[\rulename{Corrupt}]
    	{\shoare
    		{(\forall v \ldotp \ext{\ecorrupt{e}{v}}{\snapshot}) \land (\forall p, v \ldotp P[\snapshot + p + \ecorrupt{e}{v} / \snapshot])}
    		{\ccorrupt{e}}
    		{P}}

    \\[2em]
    
    \Inf[\rulename{Fork}]
    	{\freevars{R} \subseteq \vec{x}}
    	{P \models R}
    	{\shoare{R}{C}{\ttrue}}
    	{\shoare{P}{C'}{Q}}
    	{\shoare{P}{\cfork{\vec{x}}{C} \cseq C'}{Q}}

    \end{array}
    \]
\end{center}

\caption{The proof rules.}
\label{fig:rules}
\end{figure*}

\Figref{rules} shows the proof rules for our logic.
Our assertion language is a first-order logic (for brevity not a separation logic) with the usual logical connectives and quantifiers, and access to local program variables.
Pre- and postconditions can therefore refer to the local snapshot, but they \emph{cannot} refer to the global trace. The latter corresponds to our methodology (recall \secref{local-snapshots}), where pre- and postconditions also cannot directly express properties about the trace because access to it is governed by our library's ghost lock. Instead, properties about the global trace, such as the existence of a particular trace event, must always be expressed via the local snapshot. This ensures that pre- and postconditions are stable under potential environment interference, which is needed to prove our proof rules sound.

Similar to the discussion of the operational semantics, we discuss only non-standard proof rules.
Proof rules corresponding to commands that modify the global trace, \eg \rulename{Send}, enforce that the trace invariant is maintained under potential environment interference.
For this purpose, we define an extensibility predicate specifying that appending a trace $n$ event to an arbitrary extension of a trace~$\tau$ maintains the trace invariant.

\begin{definition}\textbf{Extensibility.}\label{def:ext}
	A trace~$\tau$ is \emph{extensible} by a trace event~$n$ if the trace invariant~\traceinvariant is maintained for any possible trace~$\tau'$, given that $\tau$ is a prefix thereof:
	\begin{equation*}
		\ext{n}{\tau} \triangleq \; \forall \tau' \ldotp \prefix{\tau}{\tau'} \land \traceinvariant(\tau') \implies \traceinvariant(\tau' + n)
	\end{equation*}
\end{definition}

Recall from \figref{semantics} that commands modifying the global trace, \eg send, also update the local snapshot to the most recent version of the trace. Analogous to the proof rule for assignments, the proof rules for trace-modifying command thus require that the syntactically substituted postcondition $\forall p \ldotp P[\snapshot + p + n / \snapshot]$ holds in the state before executing the command, where $\mathit{n}$ is a trace event (\eg $\esend{e}$).
The quantified $p$ accounts for all possible trace extensions that could have been made by the environment since the local snapshot was last updated, and thus accounts for arbitrary interleavings of participants and the attacker.

For the sake of presentation we have omitted additional assumptions that are available when discharing preconditions of snapshot-updating commands: \eg in proof rule \rulename{NonceGen} we may additionally use nonce freshness, and in proof rule \rulename{Recv} we may assume that a received message was previously sent and not dropped in the meantime.

\begin{theorem}\textbf{Soundness of proof rules.}\label{thm:soundness}
	\begin{align*}
		\text{If} \; \shoare{P}{C}{Q} \; \text{then} \; \hoare{P}{C}{Q}
	\end{align*}
\end{theorem}

We proof this theorem in the usual way, by structural induction on the shape of the proof tree given by the theorem's left-hand side of the implication. We proceed by a case distinction on the last rule applied, and may assume the theorem (\ie our induction hypothesis) for this rule's premises. In our proof sketch we focus on a few interesting cases -- send, sequential composition, and fork -- and we present these cases further down, as individual lemmas.
Send is interesting because it illustrates a trace-updating proof rule, for which we have to show that the trace invariant is maintained.
The challenge for sequential composition is to show that our definition of configuration safety allows us to prove that \emph{each} transition in the system maintains the trace invariant.
The \rulename{Fork}~proof rule is of interest because it is the only command that extends the system configuration with additional components.

We begin by sketching the proofs for several auxiliary lemmas about configuration safety that will be useful later on.

\begin{lemma}\label{lem:cempty-safety}
	The empty command satisfies configuration safety given that the predicate~$Q$ holds.
	\begin{align*}
		\forall n, i, \sigma, Q, \tau \ldotp Q(\sigma) \implies \safe{n}{i}{\cempty}{\sigma}{Q}{\tau}
	\end{align*}
\end{lemma}
\begin{proof}
	We show for arbitrary $n$, $i$, $\sigma$, $\tau$, and assuming $Q(\sigma)$, that $\safe{n}{i}{\cempty}{\sigma}{Q}{\tau}$ holds.
	Case $(i)$ from the definition of $\mathit{safe}$ holds straightforwardly.
	Case $(ii)$ is satisfied because there is no transition starting in command~$\cempty$ and resulting in a different command. Hence, this case vacously holds.
\end{proof}

\begin{lemma}\label{lem:less-safety}
	A command~$C$ satisfying configuration safety for $n$~execution steps is safe to execute for fewer execution steps.
	\begin{align*}
		& \forall m, n, i, C, \sigma, Q, \tau \ldotp m \leq n \land \safe{n}{i}{C}{\sigma}{Q}{\tau} \\
		& \implies \safe{m}{i}{C}{\sigma}{Q}{\tau}
	\end{align*}
\end{lemma}
\begin{proof}
	Straightforward induction on $m$.
\end{proof}

Next, we present soundness lemmas for the aforementioned interesting proof rules: send, sequential composition, and fork.

\mypar{\rulename{Send}}
Soundness for the proof rule~\rulename{Send} directly follows from the following safety lemma:

\begin{lemma}\label{lem:csend-safety}
	\begin{align*}
		& \forall n, i, \sigma, Q, \tau \ldotp \ext{\esend{e(\sigma)}}{\snapshot(\sigma)} \land {} \\
		& \phantom{\forall n, i, \sigma, Q, \tau \ldotp} (\forall p \ldotp Q[\snapshot + p + \esend{e} / \snapshot](\sigma)) \\
		& \implies \safe{n}{i}{\csend{e}}{\sigma}{Q}{\tau}
	\end{align*}
\end{lemma}
\begin{proof}
	We prove this lemma by induction on $n$ using the following induction hypothesis:
	\begin{align*}
		\mathit{IH}(n) \triangleq \; & \forall i, \sigma, Q, \tau \ldotp \ext{\esend{e(\sigma)}}{\snapshot(\sigma)} \land {} \\
		& \phantom{\forall i, \sigma, Q, \tau \ldotp} (\forall p \ldotp Q[\snapshot + p + \esend{e} / \snapshot](\sigma)) \\
		& \implies \safe{n}{i}{\csend{e}}{\sigma}{Q}{\tau}
	\end{align*}
	In the base case ($n = 0$), $\safe{0}{i}{\csend{e}}{\sigma}{Q}{\tau}$ holds by definition.
	For the induction step, we assume $\mathit{IH}(n)$ and show that $\mathit{IH}(n + 1)$ holds.
	\Ie we further assume $\ext{\esend{e(\sigma)}}{\snapshot(\sigma)}$ and $\forall p \ldotp Q[\snapshot + p + \esend{e} / \snapshot](\sigma)$ for arbitrary $i$, $\sigma$, $Q$, and $\tau$.
	We have to prove that $\safe{n+1}{i}{\csend{e}}{\sigma}{Q}{\tau}$ holds.
	Case $(i)$ from the definition of $\mathit{safe}$ holds trivially because $\csend{e} \neq \cempty$.
	To prove case $(ii)$, we assume the implication's left-hand side and show that the right-hand side holds.
	In particular, we consider a transition that executes command~$\csend{e}$.
	According to the operational semantics, only the transition rule~\rulename{Local} with an application of the \rulename{Send}~rule in its premise is applicable and modifies the command in the $i$th component's configuration.
	This allows us to conclude that the considered transition must have the following shape:
	$$\smallstepsystem{\systemstatetwo{\overrightarrow{\participantstate{C}{\sigma}}}{k_a}{\tau}}{\systemstatetwo{\overrightarrow{\participantstate{C'}{\sigma'}}}{k'_a}{\tau'}}$$
	where
	\begin{align*}
		& |\vec{C'}| = |\vec{C}| \land k'_a = k_a \land \tau' = \tau + \esend{e(\sigma)} \land {} \\
		& \vec{C'_i} = \cempty \land \vec{\sigma'_i} = \vec{\sigma_i}[\snapshot \mapsto \tau + \esend{e(\sigma)}] \land {} \\
		& (\forall j \ldotp i \neq j \implies \vec{C'_j} = \vec{C_j} \land \vec{\sigma'_j} = \vec{\sigma_j})
	\end{align*}
	and $\traceinvariant(\tau) \land \prefix{\snapshot(\vec{\sigma_i})}{\tau}$ holds.
	We have to prove that (1) $\traceinvariant(\tau')$, (2) $\prefix{\tau}{\tau'}$, (3) $\prefix{\snapshot(\vec{\sigma'_i})}{\tau'}$, (4) $k_a \subseteq k'_a$, and (5) $\safe{n}{i}{\vec{C'_i}}{\vec{\sigma'_i}}{Q}{\tau'}$ hold.
	Note that no additional local configurations have been added by this command because $|\vec{C'}| = |\vec{C}|$ holds.
	(1) follows directly by definition of \defref{ext}.
	(2) holds by choosing $p = \esend{e(\sigma)}$ as witness in \defref{prefix}.
	(3) holds by reflexivity (\cf \lemref{prefix-reflexivity}).
	(4) holds because the attacker knowledge is unchanged.
	Finally, (5) follows from \lemref{cempty-safety} via the following derivation to obtain $Q(\vec{\sigma'_i})$:
	\begin{align*}
		& \forall p \ldotp Q[\snapshot + p + \esend{e}/\snapshot](\vec{\sigma_i}) \\
		& \phantom{\forall} \implies Q[\tau + \esend{e}/\snapshot](\vec{\sigma_i}) \\
		& \phantom{\forall} \iff Q(\vec{\sigma_i}[\snapshot \mapsto \tau + \esend{e(\vec{\sigma_i})}]) \iff Q(\vec{\sigma'_i})
	\end{align*}
	where the implication is justified by the fact that $\prefix{\snapshot(\vec{\sigma_i})}{\tau}$ holds.
\end{proof}

\mypar{\rulename{Seq}}
In the case where the last rule applied in our proof tree is \rulename{Seq}, we may assume the induction hypothesis for the rule's premises, \ie $\hoare{P}{S_1}{R}$ and $\hoare{R}{S_2}{Q}$. Soundness for this case, \ie showing $\hoare{P}{S_1 \cseq S_2}{Q}$, then follows from the following safety lemma:

\begin{lemma}\label{lem:cseq-safety}
	\begin{align*}
		& \forall n, i, S_1, S_2, \sigma_1, R, Q, \tau \ldotp \safe{n}{i}{S_1}{\sigma_1}{R}{\tau} \land {} \\
		& \phantom{\forall i, S_1} (\forall m, \sigma_2, \tau' \ldotp m \leq n \land R(\sigma_2) \implies \safe{m}{i}{S_2}{\sigma_2}{Q}{\tau'}) \\
		& \implies \safe{n}{i}{S_1 \cseq S_2}{\sigma_1}{Q}{\tau}
	\end{align*}
\end{lemma}
\begin{proof}
	We perform induction on $n$ using the following induction hypothesis:
	\begin{align*}
		\mathit{IH}(n) \triangleq \; & \forall i, S_1, S_2, \sigma_1, R, Q, \tau, \ldotp {} \\
		& \phantom{\forall i,} \safe{n}{i}{S_1}{\sigma_1}{R}{\tau} \land {} \\
		& \phantom{\forall i,} (\forall m, \sigma_2, \tau' \ldotp m \leq n \land R(\sigma_2) \implies \safe{m}{i}{S_2}{\sigma_2}{Q}{\tau'}) \\
		& \implies \safe{n}{i}{S_1 \cseq S_2}{\sigma_1}{Q}{\tau}
	\end{align*}
	The base case ($n = 0$) holds by definition.
	In the induction step, we may assume $\mathit{IH}(n)$ to prove $\mathit{IH}(n + 1)$. For arbitrary $i$, $S_1$, $S_2$, $\sigma_1$, $R$, $Q$, and $\tau$ we assume the left-hand side, \ie $\safe{n + 1}{i}{S_1}{\sigma_1}{R}{\tau}$ and $\forall m, \sigma_2, \tau' \ldotp m \leq n + 1 \land R(\sigma_2) \implies \safe{m}{i}{S_2}{\sigma_2}{Q}{\tau'}$. 
	It remains to prove that $\safe{n + 1}{i}{S_1 \cseq S_2}{\sigma_1}{Q}{\tau}$ holds.
	The proof proceeds similarly to the proof of \lemref{csend-safety} except that in case $(ii)$ the \rulename{Local}~rule's premise is fulfilled by an application of either the \rulename{Seq1} or \rulename{Seq2}~rule:
	\begin{itemize}
		\item Case \rulename{Seq1}:
			According to this transition's premise, there exists a transition $\smallsteppart{S_1}{\participantstate{\sigma_i}{\tau}}{\cempty}{\participantstate{\sigma'_i}{\tau''}}$ for some $\tau''$.
			Thus, we obtain by definition of $\safe{n + 1}{i}{S_1}{\sigma_i}{R}{\tau}$ that $\traceinvariant(\tau'')$, $\prefix{\tau}{\tau''}$, $\prefix{\snapshot(\sigma'_i)}{\tau''}$, $k_a \subseteq k'_a$, and \\ $\safe{n}{i}{\cempty}{\sigma'_i}{R}{\tau''}$ hold.
			We distinguish two cases, namely $n = 0$ and $n > 0$.
			In the first case, we obtain by definition $\safe{0}{i}{S_2}{\sigma'_i}{Q}{\tau''}$.
			In the second case, we obtain by definition of $\safe{n}{i}{\cempty}{\sigma'_i}{R}{\tau''}$ that $R(\sigma'_i)$ holds.
			Therefore, we can instantiate $m$, $\sigma_2$, and $\tau'$ with $n$, $\sigma'_i$, and $\tau''$, respectively, in the quantifier above.
			Thus, we obtain $\safe{n}{i}{S_2}{\sigma'_i}{Q}{\tau''}$.
			This concludes the proof for both cases $n = 0$ and $n > 0$ showing that $\safe{n + 1}{i}{S_1 \cseq S_2}{\sigma_i}{Q}{\tau}$ holds.
		
		\item Case \rulename{Seq2}:
			This transition's premise specifies that a transition $\smallsteppart{S_1}{\participantstate{\sigma_i}{\tau}}{S'_1}{\participantstate{\sigma'_i}{\tau''}}$ for some $\tau''$ exists.
			We apply the definition of $\safe{n + 1}{i}{S_1}{\sigma_i}{R}{\tau}$ to obtain $\traceinvariant(\tau'')$, $\prefix{\tau}{\tau''}$, $\prefix{\snapshot(\sigma'_i)}{\tau''}$, $k_a \subseteq k'_a$, and \\ $\safe{n}{i}{S'_1}{\sigma'_i}{R}{\tau''}$.
			By applying the induction hypothesis for $n$, we obtain $\safe{n}{i}{S'_1 \cseq S_2}{\sigma'_i}{Q}{\tau''}$.
			Thus, we showed $\safe{n + 1}{i}{S_1 \cseq S_2}{\sigma_i}{Q}{\tau}$.
	\end{itemize}
\end{proof}

\mypar{\rulename{Fork}}
Soundness of the \rulename{Fork}~proof rule follows from the following safety lemma:

\begin{lemma}\label{lem:cfork-safety}
	\begin{align*}
		& \forall n, i, \vec{x}, S_1, S_2, \sigma_1, Q, \tau \ldotp \safe{n}{i}{S_1}{\sigma_1}{Q}{\tau} \land {} \\
		& \phantom{\forall n, i, \vec{x}} (\forall j, \sigma_2 \ldotp \left[\sigma_1 \sim \sigma_2\right]^{\vec{x} \cup \snapshot} \implies \safe{n}{j}{S_2}{\sigma_2}{\ttrue}{\tau}) \\
		& \implies \safe{n}{i}{\cfork{\vec{x}}{S_2} \cseq S_1}{\sigma_1}{Q}{\tau}
	\end{align*}
	where $\left[\sigma_1 \sim \sigma_2\right]^{\vec{x} \cup \snapshot}$ denotes that $\sigma_1$ maps the variables in $\vec{x}$ and variable~$\snapshot$ to the same values as $\sigma_2$ does.
\end{lemma}
\begin{proof}
	We perform induction on $n$ and use the following induction hypothesis:
	\begin{align*}
		\mathit{IH}(n) \triangleq \; & \forall i, \vec{x}, S_1, S_2, \sigma_1, Q, \tau, \ldotp {} \\
		& \phantom{\forall i} \safe{n}{i}{S_1}{\sigma_1}{Q}{\tau} \land {} \\
		& \phantom{\forall i} (\forall j, \sigma_2 \ldotp \left[\sigma_1 \sim \sigma_2\right]^{\vec{x} \cup \snapshot} \implies \safe{n}{j}{S_2}{\sigma_2}{\ttrue}{\tau}) \\
		& \implies \safe{n}{i}{\cfork{\vec{x}}{S_2} \cseq S_1}{\sigma_1}{Q}{\tau}
	\end{align*}
	For $n = 0$, $\safe{0}{i}{\cfork{\vec{x}}{S_2} \cseq S_1}{\sigma_1}{Q}{\tau}$ holds by definition.
	In the induction step, we assume $IH(n)$ to show $IH(n + 1)$. We assume the left-hand side, \ie
	\begin{align}
		& \safe{n + 1}{i}{S_1}{\sigma_1}{Q}{\tau} \land {} \label{eq:cfork-safety-conj1} \\
		& (\forall j, \sigma_2 \ldotp \left[\sigma_1 \sim \sigma_2\right]^{\vec{x} \cup \snapshot} \implies \safe{n + 1}{j}{S_2}{\sigma_2}{\ttrue}{\tau}) \label{eq:cfork-safety-conj2}
	\end{align}
	and seek to show $\safe{n + 1}{i}{\cfork{\vec{x}}{S_2} \cseq S_1}{\sigma_1}{Q}{\tau}$.
	Similar to the proof of \lemref{csend-safety}, the interesting case is $(ii)$ in which we only consider the inference rule~\rulename{Fork}.
	Based on the operational semantics, we obtain
	\begin{align*}
		& (\vec{C'_i} = S_1) \land (\sigma_1 = \vec{\sigma_i} = \vec{\sigma'_i}) \land (|\vec{C'}| = |\vec{C}| + 1) \land {} \\
		& (\vec{C'}_{|\vec{C'}|} = S_2) \land \left[\sigma_1 \sim \vec{\sigma'}_{|\vec{C'}|}\right]^{\vec{x} \cup \snapshot} \land {} \\
		& (\forall j \ldotp 1 \leq j \leq |\vec{C}| \implies \vec{\sigma'_j} = \vec{\sigma_j}) \land {} \\
		& (\forall j \ldotp 1 \leq j \leq |\vec{C}| \land i \neq j \implies \vec{C'_j} = \vec{C_j})
	\end{align*}
	Since the attacker knowledge~$k_a$ and global trace~$\tau$ remain unchanged by the application of this inference rule, we have to prove that (a) $\safe{n}{i}{S_1}{\sigma_1}{Q}{\tau}$ and (b) $\safe{n}{|\vec{C'}|}{S_2}{\vec{\sigma'}_{|\vec{C'}|}}{\ttrue}{\tau}$ hold.
	(a) follows from applying \lemref{less-safety} to $\safe{n + 1}{i}{S_1}{\sigma_1}{Q}{\tau}$.
	Since (\ref{eq:cfork-safety-conj2})'s left-hand side is satisfied for $\sigma_2 = \vec{\sigma'}_{|\vec{C'}|}$, we obtain $\safe{n + 1}{|\vec{C'}|}{S_2}{\vec{\sigma'}_{|\vec{C'}|}}{\ttrue}{\tau}$ by instantiating the quantifier~$j$ with $|\vec{C'}|$. We also obtain (b) by applying \lemref{less-safety}.
\end{proof}

\subsection{Trace Inclusion}
\label{sec:trace-inclusion}

\begin{figure}[t]
  \makeatletter
  \lst@AddToHook{OnEmptyLine}{\vspace{-0.4\baselineskip}}
  \makeatother
\begin{gobra}
func main(num_initiators, num_responders int) {
  ... // initialization code
  while (num_initiators > 0) {
    fork (initiator_args) {
      initiator(initiator_args)
    }
    num_initiators := num_initiators - 1
  }
  while (num_responders > 0) {
    fork (responder_args) {
      responder(responder_args)
    }
    num_responders := num_responders - 1
  }
  fork() { attacker() }
}
\end{gobra}
\vspace{-0.7em}
\caption{Sketch of a program~$C_{\mathit{system}}$ bootstrapping the distributed system by first executing sequential initialization code to, \eg generate public/private keypairs and then forking several instances of an initiator and responder implementation and the highly non-deterministic attacker implementation. \code{initiator_args} and \code{responder_args} are abbreviations for a list of arguments that are passed to the initiator and responder implementations, respectively. \Eg the initiator's public/private keypair and the responder's public key might constitute \code{initiator_args}.
}
\label{fig:system-code}
\vspace{-0.5em}
\end{figure}

We can now show the desired trace inclusion (recall \secref{soundness-proof-sketch}), which directly follows from \thmref{soundness}.

\begin{theorem}\label{thm:trace-inclusion}
	If we boostrap the distributed system from a single component, with no precondition, a trace invariant that holds for the empty trace, and an initial attacker knowledge set, then the trace invariant always holds, regardless of how many transitions are performed, and additional components (participants and the attacker) are forked.
	\begin{align*}
		\forall & C, \vec{C'}, Q, \sigma, \vec{\sigma'}, k'_a, \tau' \ldotp \shoare{\ttrue}{C}{Q} \land \traceinvariant(\emptytrace) \land {} \\
		& \phantom{C, \vec{C'}, Q, \sigma, \vec{\sigma'}, k'_a, \tau' \ldotp} \smallstepsystemstar{\systemstatetwo{\participantstate{C}{\sigma}}{k_a^{\mathit{init}}}{\emptytrace}}{\systemstatetwo{\overrightarrow{\participantstate{C'}{\sigma'}}}{k'_a}{\tau'}} \\
		& \implies \traceinvariant(\tau')
	\end{align*}
	where $k_a^{\mathit{init}}$ is the initial attacker knowledge consisting of all public terms.
\end{theorem}
\begin{proof}
	We prove this theorem by first applying soundness of our proof rules~(\thmref{soundness}) and expanding \defref{config-safety} because all states~$\sigma$ satisfy $\ttrue$.
	We proceed by induction over the length of transition sequences.
	Since the trace invariant holds initially, is maintained by each transition, and each command in every component of the system satisfies configuration safety, we obtain $\traceinvariant(\tau')$ for every trace~$\tau'$ that is possible after executing $n$~transitions, where $n$ is the induction variable.
\end{proof}

\balance

\Thmref{trace-inclusion} implies the following trace inclusion property where $\phi$ is a security property implied by the trace invariant~\traceinvariant, \ie $\traceinvariant \models \phi$:
\begin{align*}
	\forall C, Q \ldotp \shoare{\ttrue}{C}{Q} \land \traceinvariant(\emptytrace) \implies \mathit{Tr}(C) \subseteq \mathit{Tr}(\traceinvariant) \subseteq \mathit{Tr}(\phi)
\end{align*}
where $\mathit{Tr}(C)$ denotes the set of all traces that result from executing arbitrary many transitions according to the small-step operational semantics.
$\mathit{Tr}(\traceinvariant)$  and $\mathit{Tr}(\phi)$ are the sets of traces satisfying $\traceinvariant$ and $\phi$, respectively.
\Ie $\mathit{Tr}(\traceinvariant) = \{\tau \mid \forall \tau \ldotp \traceinvariant(\tau)\}$ and $\mathit{Tr}(\phi)$ analogously.

In our verification case studies we prove \thmref{trace-inclusion} in three~steps:
in step 1, we once-and-forall verify our reusable verification library, including a most-general attacker implementation (an iterated nondeterministic choice between all executable commands) against a partially abstract (thus sufficiently general) trace invariant. \Ie the judgement~$\shoare{true}{C_a}{true}$ we obtain for the attacker holds for all possible attackers and protocol-specific instantiations of this abstract trace invariant.
In step 2, we implement each participant in its own program $C_i$; verifying these effectively yield a judgement~$\shoare{P_i}{C_i}{Q_i}$ per participant.

\looseness=-1
In step 3, we combine these separate judgements for the protocol participants and the attacker by constructing a program~$C_{\mathit{system}}$ that first performs some sequential initialization code and then forks several instances of protocol participants and the attacker, as illustrated in \figref{system-code}.
By taking the number of participant instances as unconstraint input parameters, we obtain a result for unboundedly-many instances.
Functions and non-deterministic choices are straightforward extensions to our programming language.
The initialization code's purpose is to establish the participants' preconditions.
\Eg in our NSL case study we implement initialization code that generates public-private keypairs and passes the relevant keys to the individual protocol participants.
In the case of WireGuard, the corresponding initialization code remains an assumption, which is typical for security protocol verification and corresponds to assuming that there exists a mechanism to authentically distribute public keys.